\documentclass[12pt]{article}

\usepackage{latexsym,amssymb,amsfonts,amsmath,amsthm}
\usepackage{ytableau}
\usepackage{enumitem}
\usepackage{graphicx,color}
\usepackage{comment}

\newtheorem{theorem}{Theorem}
\newtheorem{lemma}{Lemma}
\newtheorem{corollary}{Corollary}
\newtheorem{proposition}{Proposition}
\newtheorem{remark}{Remark}
\newtheorem{definition}{Definition}

\newtheorem{assumption}{Assumption}
\newtheorem{fact}{Fact}

\numberwithin{theorem}{section}
\numberwithin{lemma}{section}
\numberwithin{corollary}{section}
\numberwithin{proposition}{section}
\numberwithin{remark}{section}
\numberwithin{definition}{section}

\newcommand{\bs}[1]{\boldsymbol{#1}}
\newcommand{\im}{\bs{\rm i}}

\newcommand{\supp}{{\rm supp}}

\newcommand{\dist}{{\rm dist}}

\usepackage{color}

\title{Comfortability of quantum walks on embedded graphs on surfaces}
\author{Yusuke Higuchi$^1$, Etsuo Segawa$^2$\\
$^1${\small Department of Mathematics,
Gakushuin University,} \\
{\small Tokyo 171-8588, Japan}
\\
$^2${\small Graduate School of Environment and Information Sciences, Yokohama National University,}\\ {\small Hodogaya, Yokohama 240-8501, Japan
}
}
\date{}

\begin{document}

\maketitle

\par\noindent
{\bf Abstract}. 
The time evolutions of discrete-time quantum walks on graphs are determined by the local adjacency relations of the graphs. 
In this paper, first, we construct a discrete-time quantum walk model that reflects the embedding on the surface so that an underlying global geometric information is reflected.  
Second, we consider the scattering problem of this quantum walk model. 
We obtain the scattering matrix characterized by the faces on the surface and detect the orientablility of the embedding using scattering information. 
For the stationary state in the scattering problem, 
the comfortability is defined as the square norm of the stationary state restricted to the internal. 
This indicates how a quantum walker is stored in the internal under the embedding. 
Then we find that a quantum walker feels more comfortable on a surface with small genus in some natural setting.
We illustrate our results with some interesting examples.

\section{Introduction}
For the scattering problems of finite graphs embedded on the surfaces, how is the underlying geometric structure estimated by a reaction to an input? 
To tackle this problem, we adopt a discrete-time quantum walk model and attempt to extract topological structures of graphs from the behavior of quantum superposition.  
Among a lot of discrete-time quantum walk models~\cite{PKS}, the Grover walk~\cite{Wa} may come to mind first. In fact, using the commutativity of the Grover matrix with any permutation matrix, the time evolution of the Grover walk is easily constructed~\cite{KroBru}, which exhibits interesting behaviors~\cite{HS2,KKSY} and also plays a key role in the quantum search~\cite{Portugal}. Moreover, there are several important properties of the Grover walk that connect not only to a random walk\cite{QCcoupling2,QCcoupling1} but also to the stationary Shr\"{o}dinger equation~\cite{KHiguchi}. 

Here, let us first apply the Grover walk as a discrete-time quantum walk model and focus on the embedding of graphs into closed surfaces $\mathbb{F}^2$ as a geometric structure. Note that any finite graph can be embedded in $3$-dimensional Euclidean space $\mathbb{R}^3$. For example, the complete graph with $4$ vertices $K_4$ has $11$ embeddings up to the homeomorphisms of the surfaces; see Figure~\ref{fig:embedding}. 
We expect that the quantum walk ``feels" these embeddings. 
Grover walk may decide whether an underlying graph admits a topological embedding on a given surface in a similar manner to the results in \cite{Corin}. On the other hand, the time evolution of Grover walk depends only on the adjacency structure of the underlying graph because 
the time evolution operator $U_{Grover}$ acts as 
 \[ U_{\mathrm{Grover}}\delta_e=\left(\frac{2}{\mathrm{deg}(t(e))}-1 \right)\delta_{\bar{e}}+\frac{2}{\mathrm{deg}(t(e))}\sum_{\epsilon:\; o(\epsilon)=t(e),\neq \bar{e}}\delta_\epsilon, \]
for any standard base labeled by arc $e$ of the underlying graph. 
This shows the local scattering at the vertex $x$ of the incident wave represented by the arc $e$~\cite{HKSS_QGW,Tanner}.
Here $o(e)$ and $t(e)$ are the origin and terminal vertices of the arc $e$, $\bar{e}$ is the inverse arc of $e$ and $\mathrm{deg}(x)$ is the degree of the vertex $x$. 
Then, the weights associated with transmitting and reflecting at the vertex $x$ are $2/\mathrm{deg}(x)$ and $2/\deg(x)-1$ in the Grover walk, which is independent of the embedding.  
Therefore, the Grover walk cannot distinguish between any two embedded surfaces of a graph. 
Thus, for our objective, we propose another type of discrete-time quantum walk model.    

To this end, we use the concept of drawing a graph without crossings of edges and extra faces on a closed surface, which is called the two-cell embedding~\cite{GT,MT,NakamotoOzeki}. 
More details are given in Section~\ref{sect.RS}. See also Figures~\ref{fig:face}, \ref{fig:oridetection} and \ref{fig:embedding}. 
It is well known that two-cell embedding on the closed surface of a graph $G=(X,E)$ can be realized using a rotation system~\cite{GT,MT,NakamotoOzeki}, where $X$ and $E$ are sets of vertices and unoriented edges, respectively. We put $A$ as the set of symmetric arcs induced by $E$.   
The rotation system is the triple of the symmetric digraph $G=(X,A)$, the rotation $\rho: A\to A$ and the twist $\tau: A\to \mathbb{Z}_2$.   
Here the rotation $\rho$ is decomposed into cyclic permutations with respect to the incoming arcs of each vertex.    
Thus our target can be switched to constructing a quantum walk model for a given rotation system $(G,\rho,\tau)$. 
We have proposed a quantum walk model on the {\it orientable} surfaces in \cite{HS3} and characterized the stationary state by some spanning subgraphs of the dual graph induced by the orientable embeddings.   
Here we will propose an extension model of \cite{HS3} that can be constructed on both orientable and non-orientable surfaces as follows.
In this paper, the abstract graph $G$ is deformed so that the information of $\rho,\tau$ is reflected, and the in- and out-degrees are equal $2$. 
The regularity of degree $2$ is derived from the implementation of optical polarizer elements~\cite{MHMHS}.  
See Sections~\ref{sect:doublecover} and \ref{sect:blowup} for a more detailed construction and Figure~\ref{fig:FigBU}. 
Such a new graph $G(\rho,\tau)$ is obtained by replacing each vertex of the double covering graph~\cite{GT,MT} induced by $\tau$ with a directed cycle induced by $\rho$.  The replaced directed cycles are called islands, and the symmetric arcs between the islands that reflect the original adjacent relation, are called the bridges. 
The tails, which are semi-infinite paths, are inserted into all the island arcs. Such an assignment of tails is called the {\it hedgehog}; see Figure~\ref{fig:FigBU} (d). 
Let us explain the time evolution below; see Sections~\ref{sect:TEF} and \ref{sect:TEIF} for further details. 
The degree-$2$ regularity of the new graph makes it possible that
the one-step time evolution on the whole space can be described by the local scattering at each vertex of the internal graph by the $2\times 2$ unitary matrix
\[ C= \begin{bmatrix} a & b \\ c & d\end{bmatrix}. \]
Here $a,b,c,d$ are the complex valued weights associated with moving of a quantum walk for one step ``from an island to the same island," `` a bridge to an island," ``an island to a bridge" and ``a bridge to the inverse bridge," respectively. The unitarity of the local coin matrix $C$ also includes the unitarity of the total time evolution operator. 
On the other hand, we set that the dynamics on the tail is {\it free}; see (\ref{eq:free}) for the detailed dynamics of the free. 
The initial state is set so that the internal graph receives constant inflow at every time step. 
If a quantum walker goes outside the interior, it never returns to the interior, owing to the free dynamics of the tails.
Such a quantum walker can be regarded as outflow. 
By balancing the in- and out-flows, this quantum walk model converges to a fixed point in the long time limit~\cite{FelHil1,FelHil2,HS}.

Now, for a fixed abstract graph $G_o$, let us vary its embedding onto the surfaces.
For every $\rho$ and $\tau$, set $\psi_\infty$ as the stationary state and $\psi_\infty|_{G_o}$ as the restriction of $\psi_\infty$ on the internal graph $G_o$. 
In this study, we focus on extracting the embedding structures of a graph from the stationary state of this quantum walk model. 
To this end, we divide the stationary state into external and internal parts; we characterize the external part by the {\it scattering matrix}, which shows the reaction of the internal graph to an input at the tails, while the internal part by the {\it comfortability}, which corresponds to the energy in the internal graph. 
We estimate how quantum walker \\

\noindent (i) gives a {\it scattering} on the embedding by finding the expression of the scattering matrix 
\[ \bs{\beta}_{out}=S\bs{\alpha}_{in}, \]
where $S$ is independent of inflow $\bs{\alpha}_{in}$ and outflow $\bs{\beta}_{out}$ and a unitary operator~\cite{FelHil1,FelHil2},  and \\

\noindent (ii) feels {\it comfortable} to the embedding by defining 
\[ \mathcal{E}(G_o,\rho,\tau):=\mathcal{E}:=\frac{1}{2} ||\;\psi_\infty|_{G_o}\;||^2, \]
that is, the larger $\mathcal{E}$ is, the more {\it comfortable} quantum walker feels to the embedding. \\ 

\noindent Throughout this paper, in addition to the hedgehog tail attachment, we impose the following assumption (2). 
\begin{assumption}\label{ass:hd}
\noindent 
\begin{enumerate}
    \item the assignment of tails is the hedgehog, more precisely see (\ref{eq:def:hedgehog});
    \item the $(2,2)$ element of $C$ $(=d)$ is a real number.
\end{enumerate}
\end{assumption}
\noindent Under such a construction of the quantum walk reflecting the embedding on a closed surface, now we are ready to state our main result on the scattering and comfortablity. 
\begin{theorem}[Scattering]\label{thm:scattering}
Assume $d\in \mathbb{R}$ and set $\omega=-\det C$ and the hedgehog assignment of the tails.
Let $F$ be the set of faces induced by the rotation system $(G,\rho,\tau)$. 
The scattering matrix is decomposed into the following $|F|$ unitary matrices as follows:  \[S=\bigoplus_{f\in F} S_{f}, \]
where 
\[ S_{f}=bc\omega P_{f}\;(I_{f}-a\omega P_{f})^{-1}+dI_{f}. \]
Here 
$I_f$ is the identity matrix and $P_f$ is a unitary and weighted permutation matrix corresponding to the closed walk on the boundary of face $f$ defined in (\ref{eq:hedgehogP}).
\end{theorem}
\noindent Theorem~\ref{thm:scattering0} in Section~\ref{sect:scat} gives the scattering matrix in a more general setting. 
The operator $P_f$ is a weighted permutation matrix induced by the closed walk along the boundary of face $f$, which is called the facial walk of $f$. 
The rotation system $(G,\rho,\tau)$ determines the set of faces $F$ and provides a unique orientation for the closed walk along the boundary of face $f\in F$. This closed walk corresponds to the facial walk, which is expressed by a sequence of arcs, and we simply write $f$ for this closed walk. 
It should be noted that $\bar{f}\not\in F$ for $f\in F$, where $\bar{f}$ is the inverse direction of $f$. 
The weight is determined by each edge type that passes through the facial walk.  
Based on this property of the scattering matrix, detection of orientability using scattering information is proposed in Theorem~\ref{thm:detection}. 

It is shown in Proposition~\ref{prop:scatiso} that the scattering matrices of two equivalent embeddings are unitarily equivalent. 
Although the scattering of the two embeddings with the same inflow is generally different, 
the comfortability $\mathbb{E}(\mathcal{E})$ under the following setting is invariant under the isomorphism of the embeddings. The reason for this is explained in Section~\ref{sect:proofComf}. 
\begin{theorem}[Comfortability]\label{thm:comfa>0}
Let $G=(X,A)$ and $(G,\rho,\tau)$ be the abstract symmetric digraph and its rotation system, respectively. Let $F$ be the set of faces determined by the rotation system $(G,\rho,\tau)$. 
Randomly pick a tail from the hedgehog tails as an input.  
Let $\mathbb{P}_1,\dots,\mathbb{P}_m$ be the hedgehog tails. 
Suppose $\bs{\alpha}_{in}$ is a random variable distributed uniformly in $\{ \delta_j \;|\; j=1,\dots,m \}$. 
Under Assumption~\ref{ass:hd} with $a>0$ and $\omega =1$,  the average of the comfortability with respect to a randomly chosen input is expressed by 
\begin{align}
\mathbb{E}[\mathcal{E}] &= 
\frac{1}{|A|} \frac{2+|b|^2}{|b|^2}\sum_{f\in F}|f| \frac{1+a^{|f|}}{1-a^{|f|}}
- \frac{1}{|A|}\frac{a}{|b|^2} \sum_{f\in F}\frac{1}{1-a^{|f|}}\sum_{e\in f\cap \bar{f}} \left(\; a^{\mathrm{dist}_f(e,\bar{e})}+a^{\mathrm{dist}_f(\bar{e},e)}\;\right). \label{eq:1st}
\end{align}
Here $f\cap \bar{f}$ is the set of self-intersections of the face $f\in F$, 
where if a facial walk $f$ passes through an arc $e$ and also its inverse $\bar{e}$, then the face $f$ is said to have a self-intersection at the unoriented  boundary edge $|e|$. 
\end{theorem}
\noindent 
Thus, the larger the first term is, the more comfortable the quantum walker feels, whereas the larger the second term is, the more uncomfortable the walker feels. 
The first term is related to an integer partition of $|A|$, whereas the second term is related to self-intersections of faces. See Figure~\ref{fig:si} for the self-intersection. 
Figure~\ref{fig:ranking} shows the ranking of the comfortability of the embeddings for the complete graph $K_4$ with $a=0.98$. 
We discuss the messages of combinatorial structures from Theorem~\ref{thm:comfa>0} in
Section~\ref{sect:cor}. 

The remainder of this paper is organized as follows. 
Section~2 describes how the geometric and combinatorial information of graph embeddings can be extracted from Theorem~\ref{thm:comfa>0} as corollaries. We show the best and worst embeddings of the complete graph $K_n$ for a quantum walker and also characterize the comfortability on the island using the Young diagram. 
In Section~3, we provide a short review of graph embeddings on orientable/non-orientable surfaces.
In Section~4, we describe the time evolution of the quantum walk induced by the rotation system $(G,\rho,\tau)$. 
In Section~5, we present the unitary equivalence of the time evolutions between the equivalent embeddings. 
In Section~6, we show that the scattering matrix is characterized by the resulting faces and is described by the direct sum of the unitary and circulant matrix matrices. Moreover, a method for detecting orientability using  scattering information is proved. 
In Section~7, we discuss the comfortability and give the proof of Theorem~\ref{thm:comfa>0}. 
\begin{figure}[hbtp]
    \centering
    \includegraphics[keepaspectratio, width=140mm]{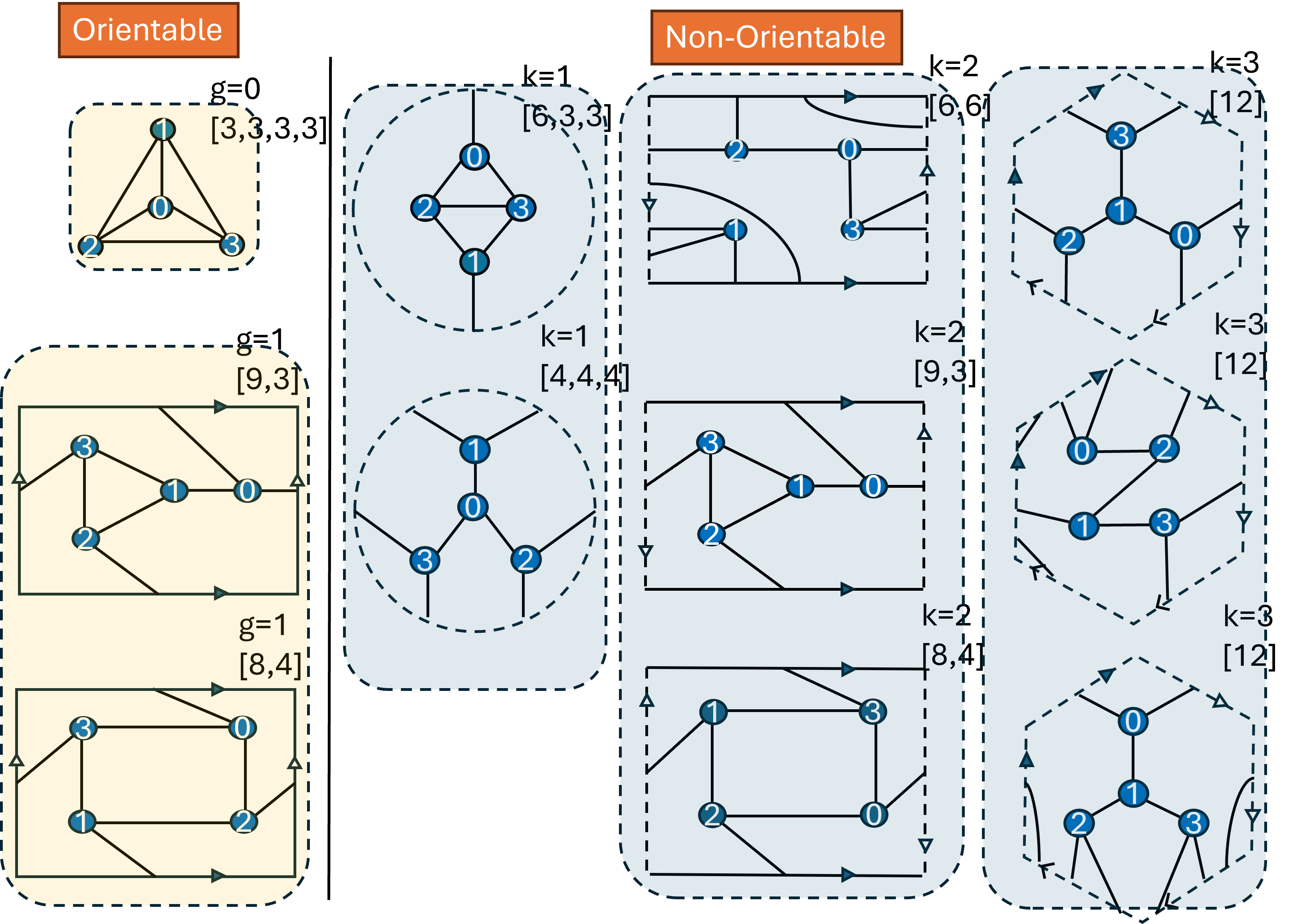}
    \caption{The list of embeddings of $K_4$: The genus is described by $g$ and $k$, for orientable and non-orientable surfaces, respectively. For example, $g=0$ and $g=1$ correspond to the surfaces of the sphere and torus, respectively, whereas $k=1$ and $k=2$ correspond to the surfaces of the projective plane and Klein's bottle, respectively. The boundary lengths of the faces of the resulting embedding are denoted as $[ \lambda_1,\lambda_2,\dots,\lambda_\kappa]$. For example, $[6,3,3]$ indicates that one hexagon and two triangles exist in the embedding. }
    \label{fig:embeddingK4}
\end{figure}
\begin{figure}[hbtp]
    \centering
    \includegraphics[keepaspectratio, width=150mm]{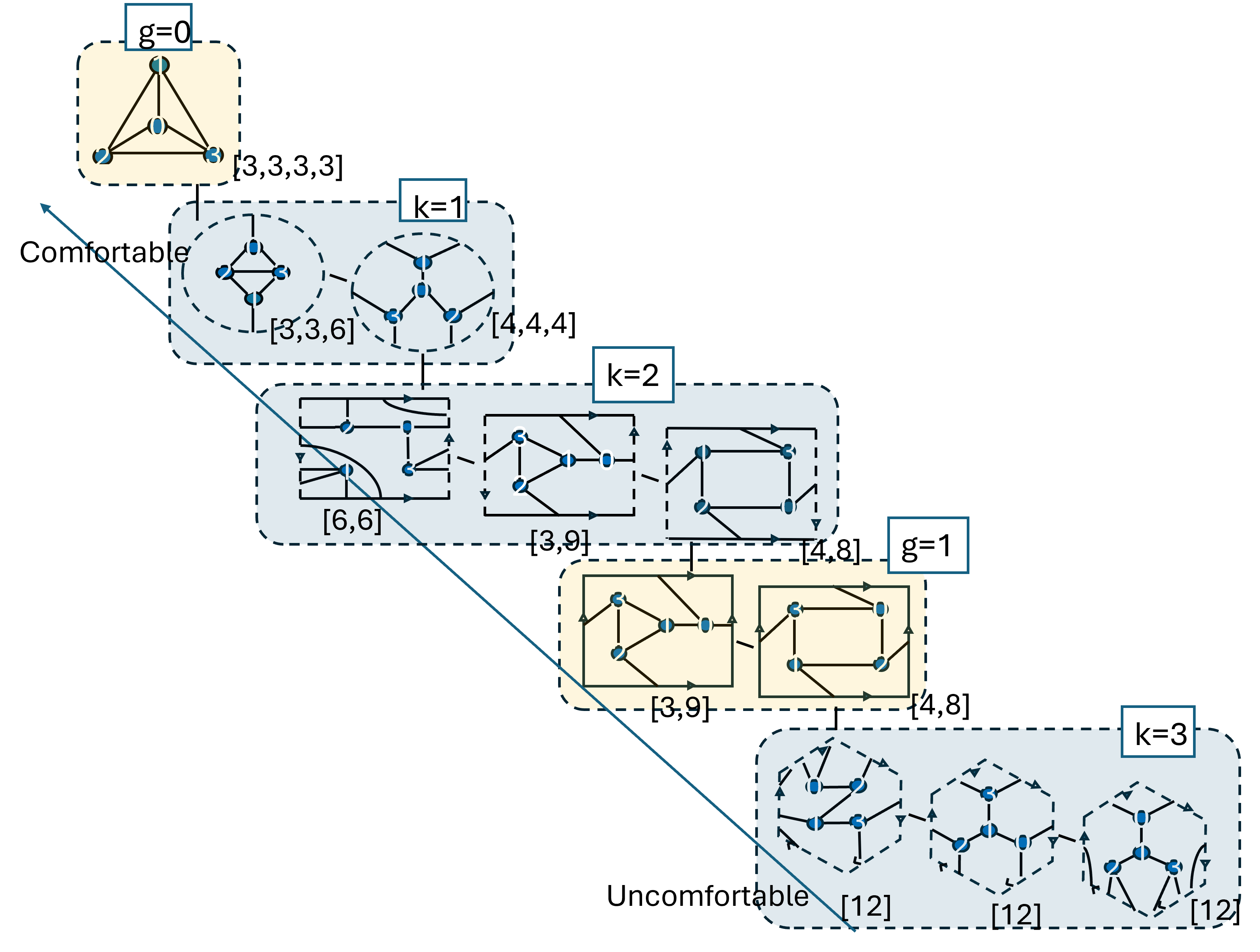}
    \caption{The ranking of the comfortability for the embeddings of $K_4$ with $a=0.98$: This ranking follows Corollary~\ref{cor:main} in the setting of $a\to 1$. The most comfortable embedding for $K_4$ is the embedding on the sphere; the worst is the embedding on the non-orientable surface with the maximum genus $k=3$. }
    \label{fig:ranking}
\end{figure}

\section{Combinatorial observations (corollaries of Theorem~\ref{thm:comfa>0})}\label{sect:cor}
In this section, we discuss 
the extraction of more detailed geometric information from
Theorem~\ref{thm:comfa>0}. 

\subsection{Observation 1: In the limit of $a\to 1$}
It is easy to observe that
if $a\to 0$, then the comfortability converges to $3$, which is completely independent of $(G,\rho,\tau)$. 
It also confirms the consistency by considering that if $a=0$, a walker is forced to trace a route `` tail$\to$island$\to$bridge$\to$island$\to$tail" by the definition of this quantum walk. 

On the other hand, if $a\to 1$, then the comfortability diverges. 
Then, taking $a=1-\delta$, we will find the appropriate scaling with respect to $\delta$, and its coefficient of the first order. We hope that the coefficient can be characterized using geometric information. Indeed, we obtain the following.  
\begin{corollary}\label{cor:main}
Let $G$ be a connected abstract graph with the vertex set $X$ and the edge set $E$. 
Let $(G,\rho,\tau)$ be the rotation system with the face set $F$.  
Under Assumption~\ref{ass:hd} with $a=1-\delta$ $(\delta\ll 1)$,  the average of the comfortability with respect to the randomly chosen initial state is expressed by 
\begin{equation}\label{eq:comfa0}
\lim_{\delta\downarrow 0}  \mathbb{E}[\mathcal{E}_\delta]\;\delta ^2=\frac{|F|}{|E|}\left( 1-\frac{1}{|F|}\sum_{f\in F} \frac{|f\cap \bar{f}|}{|f|}\right). 
\end{equation}
\end{corollary}
The smaller genus increases the number of faces, whereas the larger genus induces larger faces and then increases the number of self-intersections.
Thus, for $a\to 1$, the quantum walker feels more comfortable on the surface with a smaller genus. 
In particular, the single-face embedding on an orientable surface attains the above comfortablity $0$. 
Figure~\ref{fig:ranking} shows the ranking of the comfortability for $K_4$ following Corollary~\ref{cor:main}. 
\begin{proof}
Inserting the following expansions by small $\delta \ll 1$ into the first term of (\ref{eq:1st}),  
\[ \frac{1+a^{|f|}}{1-a^{|f|}}=\frac{1}{\delta}\frac{1}{|f|}\left\{ 2-\delta + O(\delta^2) \right\},\;\frac{2+|b|^2}{|b|^2}=\frac{1}{\delta}\left( 1+\frac{3}{2}\delta +O(\delta^2) \right), \]
we have 
\begin{align*}
\text{the first term}=\frac{1}{\delta^2}\left( \frac{|F|}{|E|}+\frac{|F|}{|E|}\delta + O(\delta^2) \right).
\end{align*}
On the other hand, inserting the following expansions by small $\delta \ll 1$ into the second term of (\ref{eq:1st}),  
\begin{align*} 
&a^{\mathrm{dist}_f(e,\bar{e})}+a^{\mathrm{dist}_f(\bar{e},e)} = 2-|f|\delta +O(\delta^2), \\
&\frac{1}{1-a^{|f|}} =\frac{1}{\delta}\frac{1}{|f|}\left( 1+\frac{|f|-1}{2}\delta+O(\delta)+O(\delta^2) \right), \\
&\frac{a}{|b|^2} = \frac{1}{2\delta}\left( 1-\frac{1}{2}\delta+(\delta^2) \right),
\end{align*}
we have 
\begin{equation*}
\text{the second term}=\frac{1}{\delta^2} \frac{1}{|E|}\sum_{f\in F} \frac{|f\cap \bar{f}|}{|f|}. 
\end{equation*}
Thus, we arrive at the desired conclusion. 
\end{proof}

\noindent {\bf Example: The best and worst embeddings of $K_n$ for quantum walker.} \\
For a fixed abstract graph, if the embedding gives the most comfortability in all the possible embeddings of the graph, then it is called the {\it best} embedding of the graph, whereas if it gives the least comfortability, then it is called the {\it worst} embedding of the graph. 
Let us determine the best and worst embeddings of the complete graph $K_n$ using the following famous graph theoretical facts: 
\begin{fact}[Minimal and maximal genera of $K_n$]\label{prop:fact}
\noindent
\begin{enumerate}
\item The (minimal) genus of $K_n$~ {\rm (Ringel and Youngs (1968) \cite{RY}):}
\begin{align*}
\text{{\rm orientable}}:\; &\gamma(K_n) =\left\lceil \frac{(n-3)(n-4)}{12}\right\rceil,  \\
\text{{\rm non-orientable}}:\; &\tilde{\gamma}(K_n) =\begin{cases}\left\lceil \frac{(n-3)(n-4)}{6}\right\rceil & \text{: $n\neq 7$,}\\
3 & \text{: $n=7$.}
\end{cases}
\end{align*}
\item The maximal genus of $K_n$ for the orientable surface~ {\rm (Nordhaus and Stewart (1971) \cite{NS}):} 
\begin{equation*}
\gamma_M(K_n)=\left\lfloor \frac{(n-1)(n-2)}{4} \right\rfloor.
\end{equation*}
\item The maximal genus of a connected graph $G=(V,E)$ for the non-orientable surface~{\rm (\cite{GT,MT} and its reference therein):}
\begin{equation}
\tilde{\gamma}_M(G)=\beta(G), 
\end{equation}
where $\beta(G)$ is the betti number of $G$, that is, $\beta(G)=|E|-|V|+1$. 
\end{enumerate}
\end{fact}
By combining Corollary~\ref{cor:main} with Fact~\ref{prop:fact}, the most comfortable embedding of $K_n$ for quantum walker can be characterized as follows. 
 \begin{corollary}\label{cor:Kn}
Under the setting of Corollary~\ref{cor:main}, 
the best and worst embeddings of the complete graph $K_n$ are as follows: 
\begin{enumerate}
 \item If $n\equiv 1,2 \mod 4$, 
\begin{itemize}
\item the best embedding is \\
any embedding on the non-orientable surface with minimal genus; 
\item the worst embedding is \\
any embedding on the orientable surface with  maximal genus.
\end{itemize}
\item If $n\equiv 0,3 \mod 4$, 
\begin{itemize}
\item the best embedding is \\
any embedding on the orientable and the non-orientable surfaces with minimal genus for $n\neq 3,4,7$; \\
any embedding on an orientable surface with a minimal genus for $n= 3,4,7$;  
\item the worst embedding is \\
embedding on a non-orientable surface  with  a maximum genus.
\end{itemize}
\end{enumerate}
\end{corollary}
The best and worst surfaces of $K_n$ for the quantum walker are listed in the following table. 
\begin{table}[hbtp]
\begin{center}
\begin{tabular}{r|ccc}
& $n\equiv 1,2\mod 4$ & $n\equiv 0,3\mod 4$, $n\neq 3,4,7$ & $n=3,4,7$ \\ \hline
Best (the minimal genus)& Non-ori & Ori and Non-ori & Ori \\
Worst (the maximal genus)& Ori & Non-ori & Non-ori
\end{tabular}
\label{table:bw}
\end{center}
\caption{The best and worst surfaces of $K_n$ for quantum walker: The genera for the best and worst surfaces are minimal and maximal genera, respectively. 
The best embedding is close to the triangulation, because $2|E|=3\{|E|-|V|+(2-(n-3)(n-4)/6) \}$ holds, where the right hand side is ``almost" $3|F|$, see (\ref{eq:EularF}). 
In particular, if $n\equiv 0,3\mod 4$ and $n\not\equiv 8,11\mod 12$, then every triangulation is the best embedding. Since every worst embedding is a single-face embedding, the comfortability of the worst embedding is $0$ for $n\equiv 1,2\mod 4$. }
\end{table}

\begin{proof}
Let us find that the best embedding and worst embedding of $K_n$ for a quantum walker are as follows. 
\begin{enumerate}
\item {\bf Best}: 
The expression (\ref{eq:comfa0}) indicates that a large number of faces (the first term) and a small number of self-intersections (the second term) make the quantum walker feel comfortable. 
Thus, we expect that a small genus embedding will be the best embedding because the small genus accomplishes both of them. 

First let us estimate the number of faces for minimum genus embedding. Combining Fact~\ref{prop:fact} with the Eular formula 
\begin{equation}\label{eq:EularF} 
|F|=|E|-|V|+2-\begin{cases} 2\gamma(K_n) & \text{: Orientable case,}\\ \tilde{\gamma}(K_n) & \text{: Non-orientable case,} \end{cases}
\end{equation} 
we have 

\begin{equation}\label{eq:exgenus}
\left\{ \,
    \begin{aligned}
    & |F^{ori}|=|F^{non-ori}|\quad  :n\equiv 0,3\mod 4 \text{ and } n\neq 3,4,7,\\
    & |F^{ori}|<|F^{non-ori}|\quad  :n\equiv 1,2\mod 4,\\
    & |F^{ori}|>|F^{non-ori}|\quad  :n=3,4,7.
    \end{aligned}
\right.
\end{equation}

Secondly, let us check whether the minimal genus embedding has a self-intersection. 
In the following, let us confirm that 
\begin{center}
there are no self-intersections in the minimum genus embeddings on both ``non-orientanble surfaces" and ``orienable surfaces except $n\equiv 2,5\mod 12$".
\end{center}
Let $n$, $m=n(n-1)/2$ and $\ell$ be the numbers of vertices, edges, and faces for the resulting embedding of $K_n$, respectively. 
Assume that there is a face with self-intersection in the embedding of $K_n$. 
Note that the boundary length of a face with a self-intersection must be at least $8$. 
Then the following inequality holds.
\begin{equation}\label{eq:start} 
2m\geq 3\times (\ell-1)+8\times 1. 
\end{equation}
\begin{enumerate}
\item Non-orientable case: 
Let $k$ be the genus of the underlying closed surface of an embedding with self-intersections. 
By the Euler formula and (\ref{eq:start}), we have 
\begin{align}
k &= 2-n-m+\ell \notag \\
&\geq 2-n+\frac{m}{3}+\frac{5}{3} \notag\\
&= \frac{1}{6}\left\{ (n-3)(n-4)+10 \right\}.
\label{eq:orieva}
\end{align}
This inequality is equivalent to
\begin{equation*}
k\geq \begin{cases}
\tilde{\gamma}(K_n)+10/6 & \text{: $n\equiv 0,1,3,4\mod 6$ and $n\neq 7$,}\\
(\tilde{\gamma}(K_n)-4/6)+10/6 & \text{: $n\equiv 2,5\mod 6$,}\\
(\tilde{\gamma}(K_n)-1)+10/6 & \text{: $n=7$,}
\end{cases}
\end{equation*}
which implies 
\begin{equation}\label{eq:k}
k>\tilde{\gamma}(K_n)
\end{equation} 
and embeddings on non-orientable surfaces with the minimal genus have no self-intersections.  
Subsequently, the second term of (\ref{eq:comfa0}) for the minimal genus embedding on the non-orientable surface is reduced to $0$. 
\item Orientable case: 
By the Euler formula for the orientable case, by replacing $k$ with $2g$ in (\ref{eq:orieva}), 
\begin{equation}
2g\geq \frac{1}{6}\left\{ (n-3)(n-4)+10 \right\}, 
\end{equation}
which is equivalent to 
\begin{equation*}
2g\geq \begin{cases}
2\gamma(K_n)+10/6 & \text{: $n\equiv0,3,4,7\mod 12$}\\
(2\gamma(K_n)-1)+10/6 & \text{: $n\equiv 1,6,8,10\mod 12$}\\
2\gamma(K_n) & \text{: $n\equiv 2,5\mod 12$}\\
(2\gamma(K_n)-4/6)+10/6 & \text{: $n\equiv 8,11\mod 12$.}
\end{cases}
\end{equation*}
Then we have 
\begin{equation}\label{eq:g}
g>\gamma(K_n)
\end{equation}
with the exception of $n\equiv 2,5\mod 12$.  
Thus, it is ensured that there are no interactions in the minimum genus embedding, except $n\equiv 2,5\mod 12$. 
Thus, the second term of (\ref{eq:comfa0}) for minimal genus embedding on the orientable surface is reduced to $0$ except $n\equiv 2,5\mod 12$. 
\end{enumerate}
Thus combining (\ref{eq:exgenus}) with (\ref{eq:k}) and (\ref{eq:g}), we obtain that the best embedding of the comfortability is the minimum genus embedding on 
\begin{equation}\label{eq:bestKn}
\begin{cases}
\text{both orientable and non-orientable surfaces} & \text{: $n\equiv0,3\mod 4 $ and $n\neq 3,4,7$}\\
\text{non-orientable surface} & \text{: $n\equiv 1,2\mod 4$}\\
\text{orientable surface} & \text{: $n=3,4,7$.}
\end{cases}  
\end{equation}
 
\item {\bf Worst}: 
The expression (\ref{eq:comfa0}) tells us that the small number of faces (the first term) and the large  number of self-intersections (the second term) make the quantum walker feel uncomfortable. 
Thus we expect that the large genus embedding will be the worst embedding because the large genus accomplishes both of them. 

The number of faces of the maximal genus of orientable and non-orientable surfaces are 
\[ |F_M^{ori}|=\begin{cases} 1 & \text{: $n\equiv 1,2\mod 4$}\\ 2 & \text{: $n\equiv 0,3\mod4$} \end{cases}\text{ and } |F_N^{non-ori}(K_n)|=1,\]
respectively. 
\begin{enumerate}
\item $n\equiv 1,2\mod 4$ case:
If the underlying surface is orientable, then there is only one face, and every boundary face has a self-intersection, which means that comfortability is reduced to $0$ by (\ref{eq:comfa0}).  
On the other hand, if the underlying surface is non-orientable, then there must exist at least one edge with no self-intersection in the face, which implies that the comfortability is non-zero by (\ref{eq:comfa0}). 
Thus the worst embedding is the maximal genus embedding on the orientable surface. 
\item $n\equiv 0,3\mod 4$ case:
If the underlying surface is orientable, then the number of faces is $2$. Then the boundary edges of the two faces have no self-intersections. 
If a twist is inserted into one of the edges of the boundary edges so called the edge twisted surgery~\cite{GT,MT}, then the two faces are merged into a single face, conserving the self-intersected edges, and the resulting surface becomes non-orientable. 
This single-face embedding on the non-orientable surface is worse than the embedding with double-face on the orientable surface. 
\end{enumerate}
The worst embedding of the comfortability is the maximal genus embedding on 
\begin{equation}\label{eq:worstKn}
\begin{cases}
\text{non-orientable surfaces} & \text{: $n\equiv0,3\mod 4 $}\\
\text{orientable surface} & \text{: $n\equiv 1,2\mod 4$}
\end{cases}  
\end{equation}
\end{enumerate}
Therefore, (\ref{eq:bestKn}) and (\ref{eq:worstKn}) lead to the desired conclusion. 
\end{proof}
\subsection{Observation 2: Comfortability on the island.}
The comfortability on the island $\mathcal{E}|_{A_{is}}$ is proportional to the first term of (\ref{eq:1st}), 
which will be shown in the proof of Theorem~\ref{thm:comf}.
Then, let us estimate the first term of (\ref{eq:1st}) using a combinatorial method. 

Set \[h(x)=x \frac{1+a^x}{1-a^x} \]
for $0<|a|<1$. 
Here we define $h(0)$ as 
\[ h(0)=\lim_{x\downarrow 0}h(x)=\frac{2}{\log|a|}. \]
Let $F=\{f_1,\dots,f_r\}$, with $|f_1|\geq |f_2|\geq \cdots \geq |f_r|$ be the set of underlying faces, where $r=|F|$. 
Then the first term of (\ref{eq:1st}) in Theorem~\ref{thm:comfa>0} can be reexpressed by
\[ \frac{1}{|A|}\;\frac{2+|b|^2}{|b|^2}\sum_{x\in\{|f_1|,\dots,|f_r|\}}h(x).\]
Note that $|A|=|f_1|+\cdots+|f_r|$. Then the boundary lengths of $F$ provide an integer partition, which is bijective to the Young diagram. 
Thus in the following, we consider the integer partition $\lambda \vdash |A|$ which increases the first term. 
Let us summarize the important properties of $h(x)$ for the above consideration as follows. 
This proof can be obtained by straightforward calculations.
\\

\noindent {\bf Properties of $h(x)$:}
\begin{enumerate}
\item For $\ell,m\in\mathbb{N}$,
\[ h(\ell+m)<h(\ell)+h(m)\;(<h(\ell+m)+h(0)\;) \]
\item For $\ell_j,m_j\in\mathbb{N}$ ($j\in\{1,2\}$) with $\ell_1+m_1=\ell_2+m_2$ and $|\ell_1-m_1|<|\ell_2-m_2|$, then 
\[ h(\ell_1)+h(m_1)<h(\ell_2)+h(m_2). \]
\end{enumerate}
For $\lambda=(x_1,x_2,\dots,x_r)\vdash |A|$ with $x_1\geq x_2\geq \cdots\geq x_r$, we can impose the condition $x_r\geq 3$ because the length of each face is longer than $2$.  
For such a Young diagram $\lambda=(x_1,\dots,x_r)$, let us put $Q(\lambda)$ by 
\[ Q(\lambda)=\sum_{x\in \{x_1,\dots,x_s\}} h(x). \]
We call this the energy of all the islands. 
Define the partial order $\lambda_1,\lambda_2\vdash |A|$ as   $\lambda_1>\lambda_2$ if and only if $Q(\lambda_1)>Q(\lambda_2)$. 
According to Property (1), a large length of the Young diagram increases $Q(\cdot)$. 
According to Property (2), the bias in row size also increases $Q(\cdot)$. 
Thus for any Young diagram $\lambda\vdash |A|$, we have 
\[ [\;|A|\;]\leq \lambda\leq \begin{cases} 
[3,3,\dots,3] & \text{: $|A|=0\mod {3}$}, \\
[4,3,\dots,3] & \text{: $|A|=1\mod {3}$}, \\
[5,3,\dots,3] & \text{: $|A|=2\mod {3}$}.
\end{cases}\] 
Then we immediately obtain the following corollary.  
\begin{corollary}\label{cor:tri}
Let $G$ be a connected abstract graph. 
If there is a triangulation in the embeddings of $G$, then it is the best embedding. 
\end{corollary}
\begin{proof}
If there is a triangulation in the embeddings, then it is the most comfortable on all the islands. Moreover, since there are no self-intersections in the triangulation, the triangulation is the most comfortable embedding according to Corollary~\ref{cor:main}. 
\end{proof}
If $n\equiv 0,3,4,7\mod 12$, the minimum genus embedding of $K_n$ becomes the triangulation, because $3|F|=2|E|$ holds by Fact~\ref{prop:fact}. Thus, we can see that  Corollary~\ref{cor:tri} is consistent with Corollary~\ref{cor:Kn} for $K_n$. \\

\noindent{{\bf Example: }}
The Hasse diagram in Fig.~\ref{fig:Hasse} describes the case for $|A|=12$ using only {\bf Properties of $h(x)$}. 
\begin{figure}[hbtp]
    \centering
    \includegraphics[keepaspectratio, width=120mm]{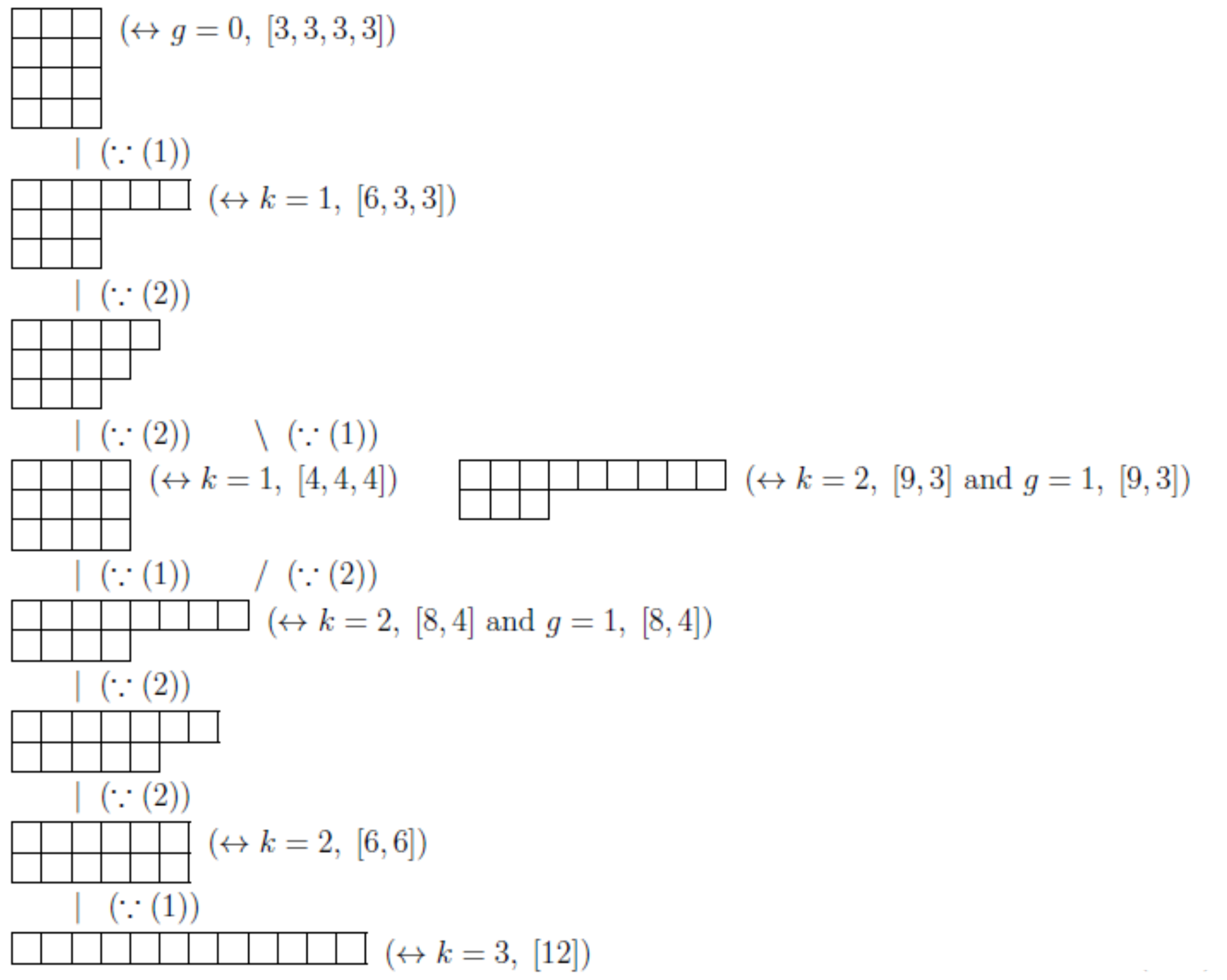}
    \caption{The Hasse diagram of the comfortability on the island of the embeddings.}
    \label{fig:Hasse}
\end{figure}
This Hasse diagram shows the partial order of the comfortability restricted to the islands for $K_4$'s embeddings. 
In this Hasse diagram, the order of the Young diagrams $[9,3]$ and $[4,4,4]$ is not determined using only {\bf Properties} (1) and (2) of $h(x)$. However
it is possible to estimate the difference between 
$Q([9,3])$ and $Q([4,4,4])$. 
In fact, we have 
$Q[9,3]<Q[4,4]+h(0)$ by the following Hasse diagram:
\begin{align*}
Q([9,3])<Q([5,4,3])<Q([4,4,3,1])<Q([4,4,4])+h(0),
\end{align*}
where all the inequalities are derived from {\bf Properties of $h(x)$} (1). On the other hand, we also have  
$Q[4,4]<Q[9,3]+h(0)$ by the following Hasse diagram:
\begin{align*}
Q([4,4,4])<Q([5,4,3])<Q([9,3])+h(0),
\end{align*}
where the first and second inequalities are derived from Properties (2) and (1), respectively. 
Then we have 
\[ |Q([9,3])-Q([4,4,4])|<h(0)=2/|\log a|. \]

Note that according to Fig~\ref{fig:ranking} or Corollary~\ref{cor:main}, the comfortability of $``k=2,\;[8,4]"$ is greater than that of $``g=1,\;[8,4]"$, while these are the same in Fig.~\ref{fig:Hasse} for all the islands. The difference derives from the number of self-intersections; there are two points of the self-intersection in the enneagon for $g=1$ while there is only one point of the self-intersection in the enneagon for $k=2$. See Figure~\ref{fig:si}. The same reason can be applied to the case for $``k=2,\;[9,3]"$ and $``g=1,\;[9,3]"$. Therefore, the non-orientability decreases the self-intersection and increases the comfortability. 
\begin{figure}[hbtp]
    \centering
    \includegraphics[keepaspectratio, width=120mm]{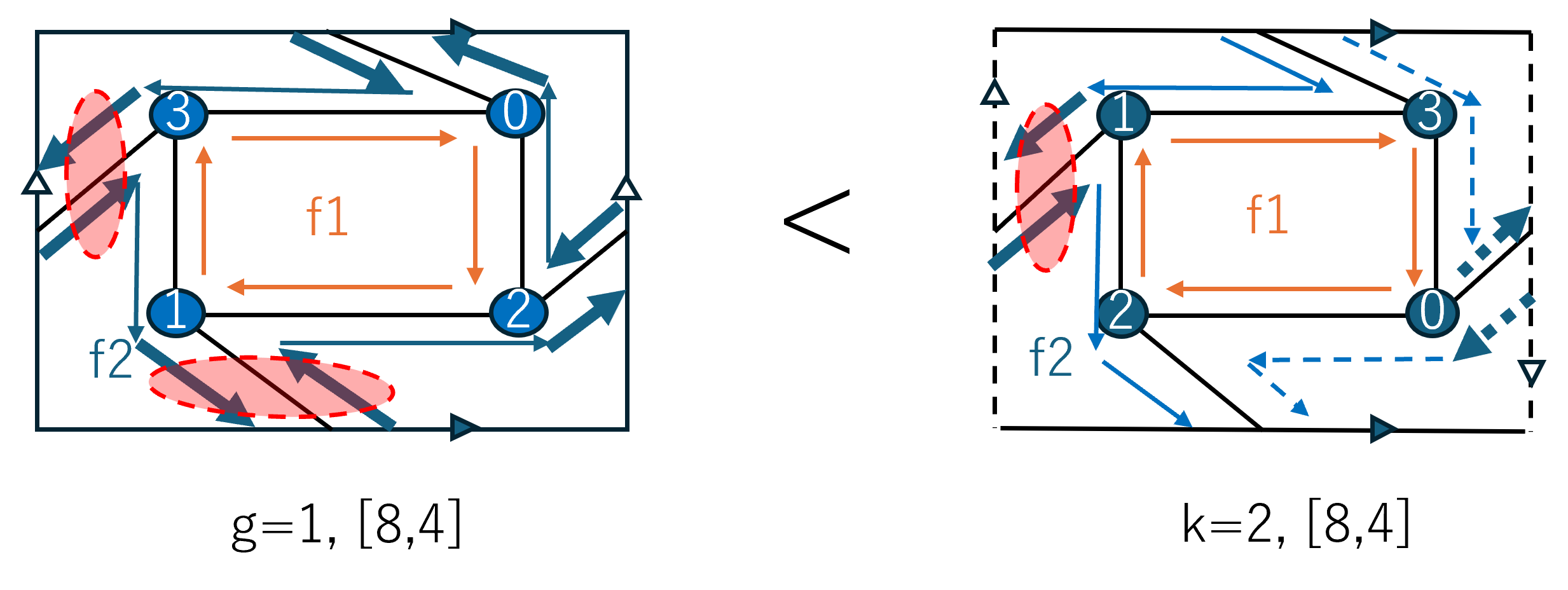}
    \caption{The self-intersection. Comparison between the embeddings on the torus and Klein bottle, whose faces are one octagon and one square. The numbers of self-intersections in the octagon at embeddings on the torus and Klein bottle are two and one, respectively. Thus the embedding on the Klein bottle is better than that on the torus in the setting of Corollary~\ref{cor:main}. }
    \label{fig:si}
\end{figure}
\section{Quick review on the rotation system}\label{sect.RS}
In this section, we give a quick review on the rotation system following \cite{GT,MT}, which will be important to construct our quantum walk model. \\

\noindent {\bf Abstract graph.}
Let $G(X,A)$ be a symmetric digraph with the set of the vertices $X$ and the set of the symmetric arcs $A$, that is, $e\in A$ if and only if $\bar{e}\in A$, where $\bar{e}$ is the inverse arc of $e$. 
The support edge of $e\in A$ is denoted by $|e|=|\bar{e}|$. 
The set of edges (which are undirected) is defined by $E=\{|e| \;|\; e\in A\}$. 
In this paper, the underlying graph is assumed to be simple without self-loops. 
For each arc $e\in A$, 
the terminus and origin of $e$ are denoted by $t(e)\in X$ and $o(e)\in X$.
The arc $e\in A$ with $o(e)=x$, $t(e)=y$ is also represented by $e=(x,y)\in A$. 
Let $A_x\subset A$ be the set of arcs whose terminal vertices are $x\in X$, that is, 
$A_x=\{e\in A \;|\; t(e)=x\}$.\\

\noindent {\bf Rotation.} 
A cyclic permutation on a countable set $W$ is 
 a bijection map $\pi: W\to W$ such that $\pi(\omega)\neq \omega$ for any $\omega\in W$. 
On each vertex $x\in X$, a cyclic permutation on $A_x$, $\rho_x: A_x\to A_x$, is assigned. The extension of $\rho_x$ to the whole arc set $A$ is given by
\[ \tilde{\rho}_x(e)=
\begin{cases}
\rho_x(e) &\text{: $e\in A_x$,}\\ 
e & \text{: otherwise.}
\end{cases}
\]
The {\it rotation} $\rho$ on the symmetric arc set $A$ is defined by 
\[ \rho=\left(\prod_{x\in X}\tilde{\rho}_x\right). \]

\noindent{\bf Twist.}
The {\it designation} for the twist of each edge is a map $\tau: A\to \mathbb{Z}_2$ such that $\tau(e)=\tau(\bar{e})$ for any $e\in A$. If  $\tau(e)=0$, the edge $e\in A$ is called {\it type-$0$}, otherwise, called {\it type-$1$}. The type-$1$ edge is regarded as twisted. \\

\noindent{\bf Rotation system.}
The triple $(G,\rho,\tau)$ is called the rotation system of $G$. 
$G$ alone is sometimes called an abstract graph. 
Let $\sigma$ be the transposition such that $\sigma(e)=\bar{e}$ for any $e\in A$.
A facial walk induced by the rotation system is the representative of the sequences of arcs identified by cyclic permutation and inverse
\[(e_0,e_1,\dots,e_{r-1}) \] 
with some natural number $r>2$   
such that 
\begin{equation}\label{eq:facialwalk}
e_{j+1}=
\begin{cases} 
\sigma (\rho (e_j)) & \text{: $\sum_{k=0}^{j}\tau(|e_j|)=0$}\\
\sigma (\rho^{-1}(e_j)) & \text{: $\sum_{k=0}^{j}\tau(|e_j|)=1$}
\end{cases}
\end{equation}
for $j=0,1,\dots,r-1$ 
in modulus of $r$. 
See Figure~\ref{fig:face} for a rotation system of $K_4$ and its facial walks. 
The set of facial walks is denoted by $F$. 
Every $f\in F$ gives the boundary of each face of the resulting two-cell embedding of the closed surface.  
Indeed, it is very starting point for our paper that the orientation system determines the embedding on a closed surface in the following meaning. 
\begin{theorem}[\cite{GT,MT,NakamotoOzeki}]
Every rotation system on graph $G$ defines (up to equivalence of embeddings) a unique locally oriented graph embedding~\cite{GT}. Conversely, every locally oriented graph embedding defines a rotation system for $G$. 
\end{theorem}
The Eular's formula gives the genus of the embedding:
\[ genus=\begin{cases} g=\frac{1}{2}(|E|-|V|-|F|)+1 & \text{: if the surface is orientable,}\\
k=|E|-|V|-|F|+2 & \text{: if the surface is non-orientable.}
\end{cases} \]
The graph $G$ contains a cycle which has odd number of type-$1$ edges if and only if the underlying surface is non-orientable. 
The following method is useful to detect the orientability of the underlying surface\cite{GT,MT}. See also Figure~\ref{fig:oridetection}: 
\\

\noindent {\bf Operation.}
\begin{enumerate}[label=\textbf{\roman*.}]
\item Choose an arbitrary spanning tree $T$ of $G$;
\item Choose a arbitrary vertex $x$ of $T$; $x$ is called the root; 
\item For any vertex which is adjacent to $x$ with the type-$1$ edge, say $y$, the orientation is reversed from $\rho_y$ to $\rho_y^{-1}$ and for any edge incident to $y$ in $G$, say $e$, the type is reversed from $\tau(e)$ to  $\tau(e)+1$; 
\item Continue this process until all the types of edges in the tree are type-$0$;
\item The surface is non-orientable if and only if there exists a type-$1$ edge in $G\setminus T$ after the above updated orientation system. 
\end{enumerate}
The following remark will play an important role in the construction of our quantum walk model. 
\begin{remark}
For any closed walk $(e_0,e_1,\dots,e_{s-1})$ on $G$, {\rm {\bf Operation   iii}} preserves the value of $\sum_{j=0}^{s-1}\tau(e_j)$, which ensures that the embedding is preserved by {\rm {\bf Opeation iii}}. 
\end{remark}
After the underlying closed surface is determined by the above procedures, for every vertex $x$, we first arrange a vertex $x$ and its ``half" edges connecting to $x$ clockwise according to the rotation $\rho_x$ on the surface. Secondly, if $x$ and $y$ are connected in $G$, then connect the corresponding half-edges each other without any crossing of edges so that the type of edge is conserved. See Figure~\ref{fig:embedding}.  
\begin{figure}[hbtp]
    \centering
    \includegraphics[keepaspectratio, width=80mm]{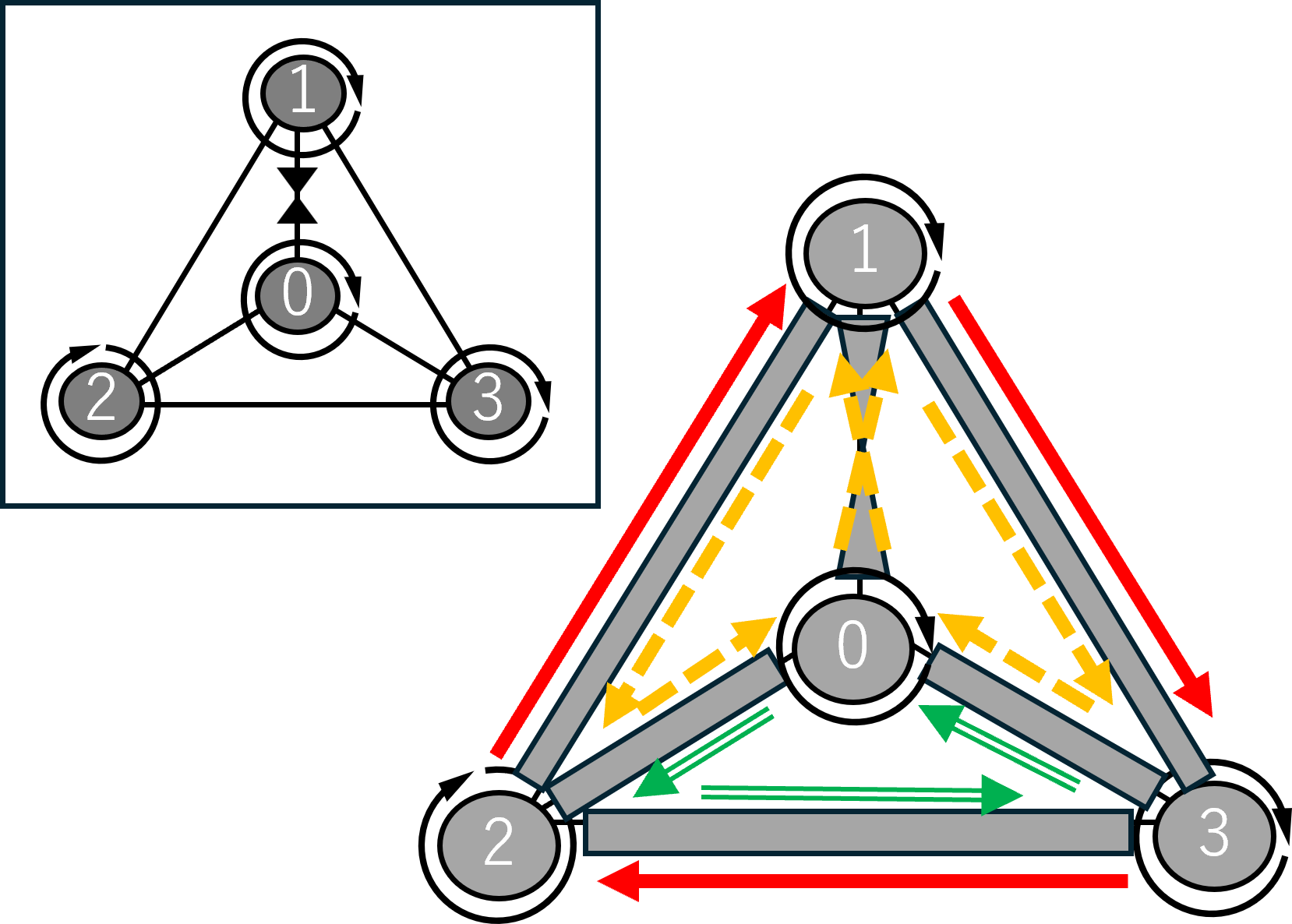}
    \caption{A rotation system of $K_4$ and facial closed walks: The rotation is assigned clockwise at each vertex, and the twist is assigned at the edge $\{0,1\}$. There are $3$ faces in this rotation system; $2$ triangles and $1$ hexagon.}
    \label{fig:face}
\end{figure}
\begin{figure}[hbtp]
    \centering
    \includegraphics[keepaspectratio, width=120mm]{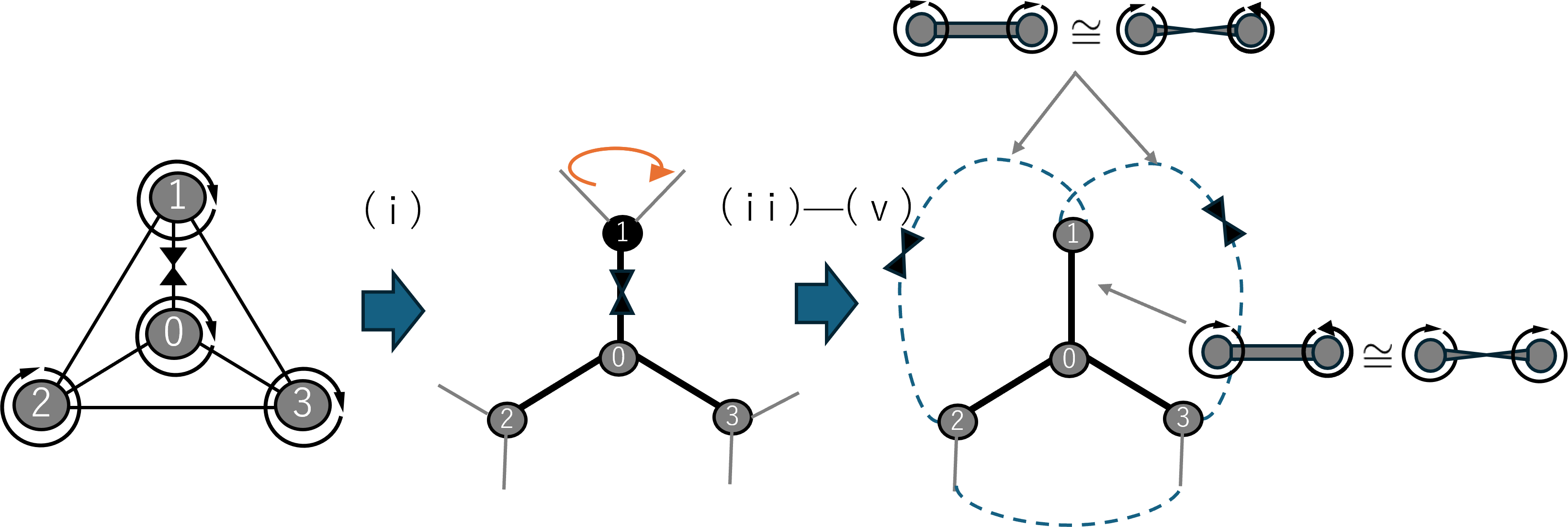}
    \caption{The detection of the orientability: In the rotation system, vertex $1$ is selected as the root of the spanning tree, and twist the vertex $1$ so that every type of edge in the tree become $0$. There are $2$ fundamental cycles containing a type-$1$ edge. Thus the closed surface must be non-orientable. By the Euler formula, the genus is $k=2-(|F|-|E|+|V|)=2-(3-6+4)=1$. Then the surface is the projective plane.  } 
    \label{fig:oridetection}
\end{figure}
\begin{figure}[hbtp]
    \centering
    \includegraphics[keepaspectratio, width=120mm]{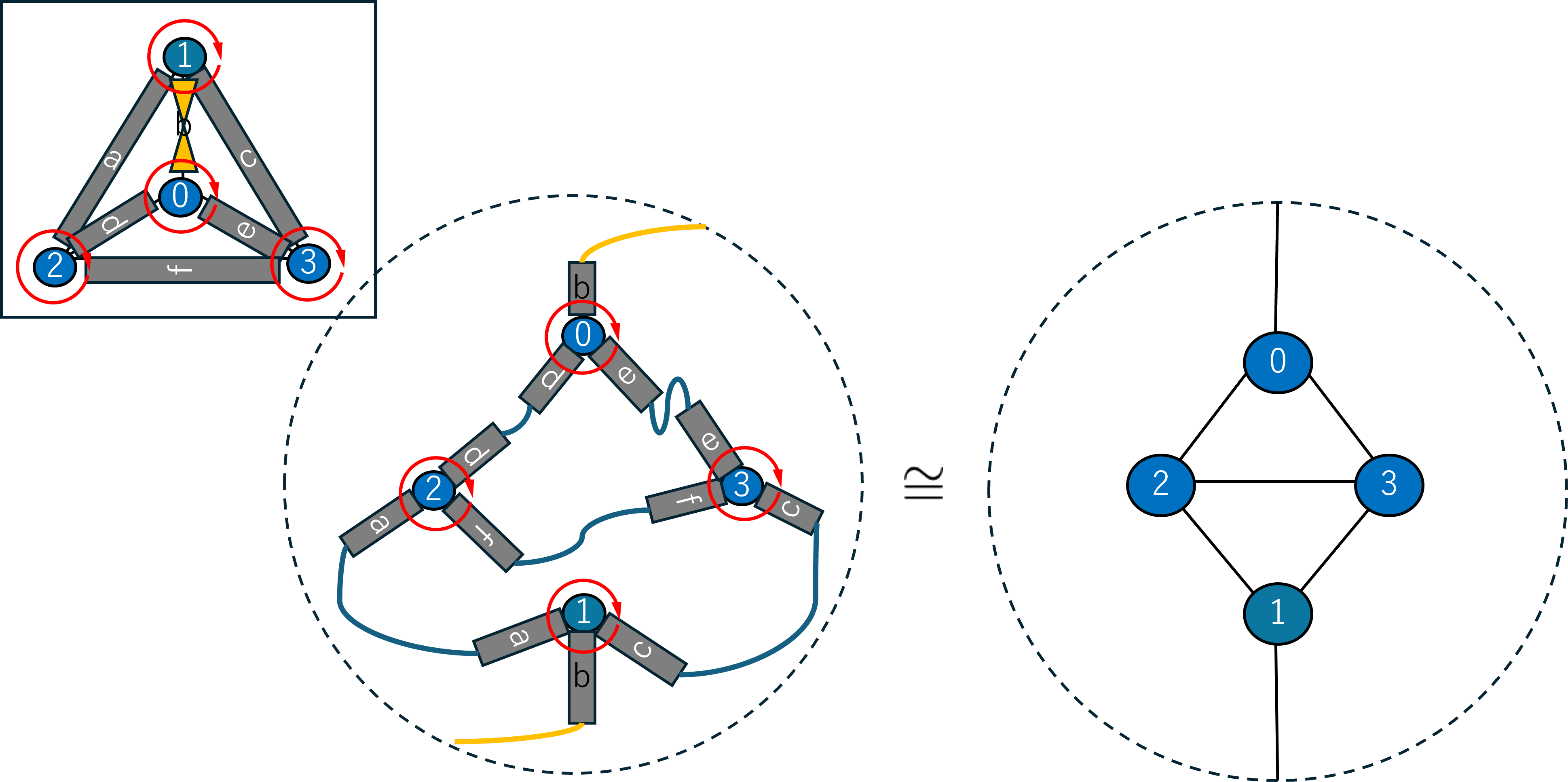}
    \caption{The drawing of $K_4$ on the projective plane: Diagonally located places on the dotted boundary are identified with each other. Incident half edges of each vertex are arranged clockwise so that its rotation is conserved. Connect corresponding half-edges without any crossing and every type-$0$ edge passes through the dotted boundary evenly, while every type-$1$ does so oddly. In this case, there are no crossings across the dotted boundary in every the type-$0$, while the type-$1$ edge crosses the dotted boundary once. }
    \label{fig:embedding}
\end{figure}

\section{Construction of quantum walk on the rotation system $(G,\rho,\tau)$}\label{sect:construction}
When a walker passes through a twisted edge on the rotation system $(G,\rho,\tau)$, the rotation of the endpoint vertex is reversed. 
To reflect the parity of the ``sheet", that is, front/back, in the dynamics of our quantum walk, let us prepare the following notions.  
\subsection{Double covering}\label{sect:doublecover}
The rotation system $(G,\rho^{-1},\tau)$ is called the chiral rotation system of $(G,\rho,\tau)$. 
To realize the parity and the rotation in the dynamics of the quantum walk model
we first prepare the two rotation systems which are chiral to each other and secondly join them by rewiring the twisted edges to the corresponding chiral vertices so that the original and its chiral rotations are conserved.  

More precisely, let us first set the double covering graph $G^\tau=(X^\tau,A^\tau)$ which is realized by the voltage assignment of $\tau:\mathbb{Z}_2\to A$ as follows.  
See also Figure~\ref{fig:FigBU} (b). 
\begin{enumerate}
\item The vertex set : $X^\tau=X\times \mathbb{Z}_2=X_0\sqcup X_1$, where $X_j=\{(x.j)\;|\; x\in X\}$ $(j\in\mathbb{Z}_2)$. 
\item The arc set : $(x,j)\in X^\tau$ and $(y,k)\in X^\tau$ are adjacent each other in $G^\tau$ if and only if $x$ and $y$ are adjacent in $G$ and $k=j+\tau(\;(x,y)\;)$. 
\end{enumerate}
Secondly, to reproduce the rotation $\rho$, the rotation $\rho$ is assigned to $X_0$, while the inverse rotation $\rho^{-1}$ is assigned to $X_1$. 
Such a resulting rotation system is called the {\it double covering} of $(G,\rho,\tau)$, which coincides with the rotation system    $(G^\tau,\rho\oplus \rho^{-1},\mathrm{id})$. 

It is possible to draw the abstract graph $G^\tau$ in the following way: 
\begin{align*}
\text{($*$)} &\text{ $X_0$ is placed in the ``left" } \\
&\text{ while $X_1$ is placed in the ``right".}
\end{align*}
\noindent We call the regions of the subgraphs in the left and right, the front and back sheets, respectively. 
Note that every vertex in the front (resp. back) sheet follows the rotation $\rho$ (resp. $\rho^{-1}$), respectively. Every arc crossing the boundary of the sheets $(\; (x,j),(y,k) \;)$ satisfies  $(x,y)\in A$ and $k\neq j$. 
See Figure~\ref{fig:FigBU} (b). 
\begin{lemma}\label{lem:switch}
Under the drawing ($*$) of $G^\tau$, 
{\rm Operation {\bf iii}} to a vertex $y\in X$ in $(G,\rho,\tau)$ is equivalent to pulling $(y,0)$ and its chiral $(y,1)$ to their opposite sheets crossing the boundary and conserving the locations of the other vertices of $G^\tau$, and switching the labeling of each vertex. 
\end{lemma}
\begin{proof}
Note that the situation of a vertex $y$ with its connected edges $D_y=\{e_1,\dots,e_r\}$ in $(G,\rho,\tau)$ is equivalent to the situation in the drawing ($*$) of $G^\tau$ that all the edges of $D_y$ and its chiral edges of $D_{y'}$ with $\tau(e_j)=0$ are in the front sheet and back sheet, respectively, while all the edges with $\tau(e_j)=1$ cross the boundary. 
Operation {\bf iii} to a vertex $y$ switches the rotation and types of all its incoming and outgoing arcs. 
Then in the situation of $G^\tau$ with the drawing ($*$), the vertex $y$ and its chiral vertex $y'$ are pulled to the opposite sheets conserving the locations of the other vertices.
After this switching, 
the resulting drawing of $G^\tau$, the rotations of $y$ and $y'$, and also the types of all the connected arcs are switched. 
\end{proof}
\begin{proposition}
The rotation system $(G,\rho,\tau)$ is non-orientable if and only if $G^\tau$ is connected. 
\end{proposition}
\begin{proof}
Every orientable rotation system can be drawn without any twisted edges, $(G,\rho,\mathrm{id})$, by Operation {\bf i}--{\bf iv}. 
The rotation system $(G,\rho,\mathrm{id})$ is orientable and the corresponding double covering graph is constructed by the two connected components, that is, disconnected because there are no type-$1$ edges.  
The operation reversing the rotation of a vertex and the types of its connected edges keeps the disconnectivity of the double covering graph by the previous lemma.  
\end{proof}

\subsection{Blow-up}\label{sect:blowup}
The blow-up graph of $G=(V,A)$ with the rotation $\rho$ is obtained by replacing each vertex of $G$ into the directed cycle following the rotation assigned to each vertex and conserving the adjacency relation of the original graph as follows. 
The vertex set of the blow-up graph $G^{BU}(\rho)=(X^{BU}(\rho), A^{BU}(\rho))$ is defined by
\[ X^{BU}(\rho) = A; \]
the arc set of the blow-up graph is defined by the disjoint union of 
\[A^{BU}(\rho)=A_{is}\sqcup A_{br},\] 
which are called the bridge and the island, respectively.  
Here the ``bridge" $A_{br}$ and the ``island" $A_{is}$ are both isomorphic to $A$ and defined by 
\[
A_{is}=\{ (e,\rho(e))\;|\; e\in A\},\;
A_{br}= \{ (e,\bar{e})\;|\; e\in A \} \]
through the direct product $A\times A$. 
See Figure~\ref{fig:FigBU} (c). 

Let us give some remarks and introduce some new notations which will be important to describe the time evolution of our quantum walk model. 
\begin{remark}
It holds that, for $e\in A$ with $t(e)=x\in X$ in $G$, 
\[ \{ ( e',\; \rho(e'))\;|\; \text{$t(e')=x$ in $G$}  \}=\{ (\;\rho^{j}(e),\; \rho^{j+1}(e)\;) \;|\; j=0,1,\dots,\deg_G(x)-1\}\subset A_{is} . \]
\end{remark}
We call the above directed cycle induced by $x\in X$ the {\it island of $x\in X$}, which is denoted by $A_{is}(x)$. 
On the other hand, for $\epsilon\in E$ with the end vertices $x,y\in X$ in $G$, $(e,\bar{e})$ and $(\bar{e},e)$, where $|e|=\epsilon$, are called the {\it bridge between the islands $x$ and $y$}. 
Then we define an extension of $\rho$ and $\tau$ to the blow-up graph $G^{BU}(\rho)$ as follows:  
\[\rho(\bs{e}_{is}):=(\rho(e),\rho^2(e)) \] 
for an island arc $\bs{e}_{is}=(e,\rho(e))\in A_{is}$,  
and  
\[ 
\tau(\bs{e}_{br}):=\tau(e) \]
for a bridge arc $\bs{e}_{br}=(e,\bar{e})\in A_{br}$. 

\noindent For each vertex $e\in X^{BU}$, the incoming arcs to $e$ come from the vertex $\rho^{-1}(e)$ along the island of $t_G(e)$ and the vertex $\bar{e}$ along the bridge between the islands $t_G(e)$ and $o_G(e)$ while the outgoing arcs from $e$ go out to the vertex $\rho(e)$ along the island of $t_G(e)$ and the vertex $\bar{e}$ along the bridge between the islands $t_G(e)$ and $o_G(e)$. 

Then for any island arc $\bs{e}_{is}\in A_{is}$, 
there are exactly two bridge arcs incoming to $\bs{e}_{is}$: $\bs{e}_{br}$ and $\bs{e}_{br}'\in A_{br}$ satisfy  $o(\bs{e}_{is})=t(\bs{e}_{br})$ and $t(\bs{e}_{is})=t(\bs{e}_{br}')$, respectively. 
Such bridge arcs for the island arc $\bs{e}_{is}$ are denoted by 
\begin{equation}\label{eq:brbr}
\bs{e}_{br}:=\mathrm{br}(\bs{e}_{is}),\;
\bs{e}_{br}':=\mathrm{br}^{\sharp}(\bs{e}_{is}), 
\end{equation} respectively. 
On the other hand, for any bridge arc $\bs{e}_{br}\in A_{br}$, there are exactly two island arcs incident to $o(\bs{e}_{br})$: $\bs{e}_{is}$ and $\bs{e}_{is}'$ satisfy $o(\bs{e}_{br})=t(\bs{e}_{is})=o(\bs{e}_{is}')$. Such island arcs for the bridge arc $\bs{e}_{br}$ are denoted by 
\begin{equation}\label{eq:isis}
\bs{e}_{is}:=\mathrm{is}(\bs{e}_{br}),\;
\bs{e}_{is}':=\mathrm{is}^{\sharp}(\bs{e}_{br}),
\end{equation}
respectively. 
See Fig.~\ref{fig:isbr}.
\begin{remark}
The in-degree and out-degree for every vertex of the blow-up graph are equally $2$.  
This condition is suitable for the purpose to implement a quantum walk model as a circuit of optical polarized elements~\cite{MHMHS}. 
\end{remark}
\begin{figure}[hbtp]
    \centering
    \includegraphics[keepaspectratio, width=180mm]{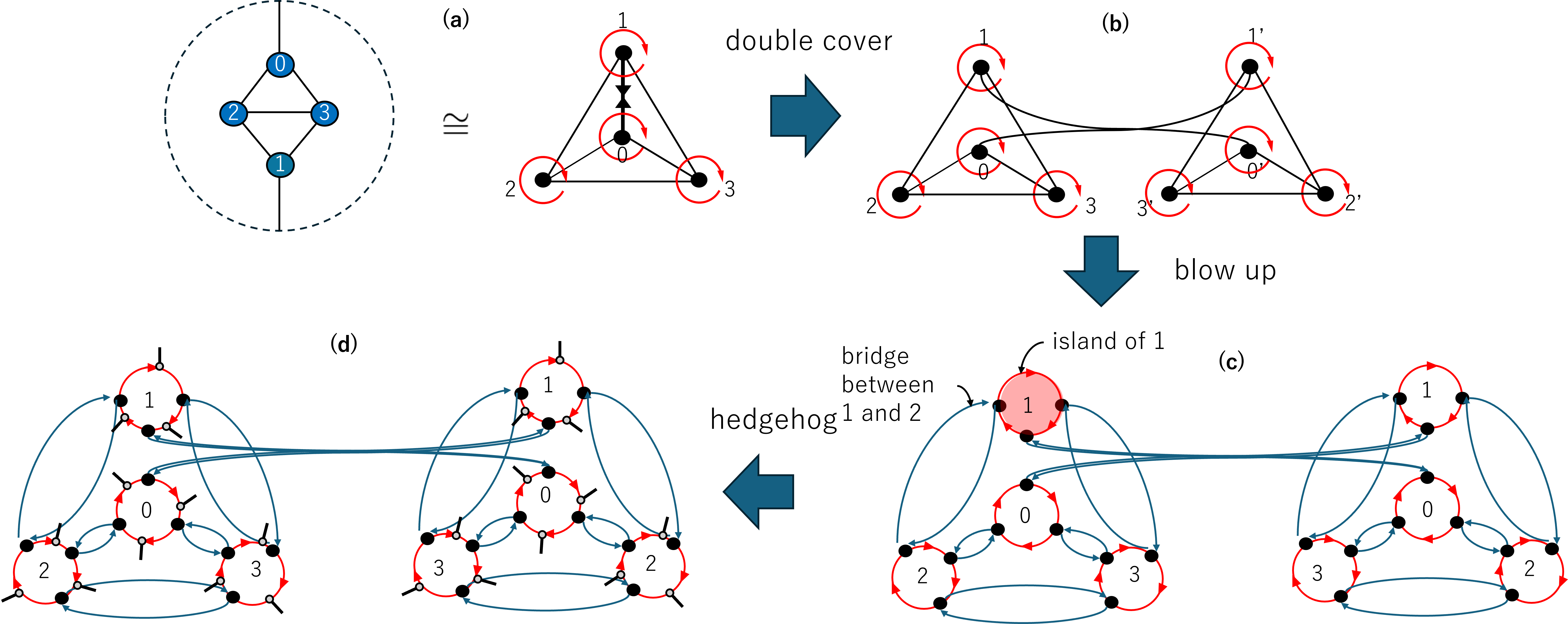}
    \caption{The construction of $G(\rho,\tau)$: (a) the rotation system of $K_4$ embedding in the projective plain $(G,\rho,\tau)$; (b) the rotation system of the double covering graph $G^\tau$, $(G^\tau,\rho\oplus\rho^{-1},\mathrm{id})$; (c) the blow-up graph of $G(\rho,\tau)$; (d) the blow-up graph with tails $\tilde{G}(\rho,\tau)$ for $\delta A_{is}=A_{is}$ case  (every tail is inserted on each island arc, and such a tail arrangement is called hedgehog.)    } 
    \label{fig:FigBU}
\end{figure}

\subsection{Quantum walk on the rotation system $(G,\rho,\tau)$}\label{sect:TEF}
For the rotation system $(G,\rho,\tau)$, let us consider the blow-up graph of $G^\tau$ induced by the rotation $\rho\oplus \rho^{-1}$. 
This blow-up graph is denoted by $G(\rho,\tau)=(X(\rho,\tau), A(\rho,\tau))$, that is,  
\[G(\rho,\tau):=(G^\tau)^{BU}(\rho\oplus \rho^{-1}).\] 

Now we are ready to describe the dynamics of this quantum walk model. The total state space of our quantum walk is given by the vector space 
\[ \mathcal{H}_{total}=\mathbb{C}^{A(\rho,\tau)}. \]
In our quantum walk model, a local time evolution at each vertex of $G(\rho,\tau)$ is represented by a $2$-dimensional unitary matrix, because the degree of the blow-up graph is $2$, and if a quantum walk passes through the twisted bridge, the phase of the quantum walker is converted by $e^{i\pi}=-1$.  

In the following,  we represent this dynamics of our quantum walk model in more precisely.  
Let $W: \mathbb{C}^{A(\rho,\tau)}\to \mathbb{C}^{A(\rho,\tau)}$ be the unitary time evolution operator such that    
for each time step $t\in \mathbb{N}$, the total state of the quantum walk at time $t$,  $\psi_t\in \mathbb{C}^{A(\rho,\tau)}$, is given by 
\[ \psi_{t+1}=W\psi_t \text{ for $t\geq 0$}\]
with some initial state $\psi_0\in \mathbb{C}^{A(\rho,\tau)}$. 
The unitary time evolution operator $W$ is defined as follows. 
See also Figs.~\ref{fig:isbr} and \ref{fig:LocTE}. 
\begin{definition}\label{def:TE1}
Let us set a $2\times 2$ unitary matrix by 
\[ C=\begin{bmatrix} a & b \\ c & d \end{bmatrix}. \]
Let $\psi_t\in \mathbb{C}^A$ be the state of the quantum walk at time $t=0,1,2,\dots$.  
The time evolution operator $W$ is defined by 
$\psi_{t+1}=W\psi_t$ such that 
for any bridge arc $\bs{e}_{br}\in A_{br}\subset A(\rho,\tau)$, the local time evolution on the vertex $o(\bs{e}_{br})\in X(\rho,\tau)$ is described by
\begin{equation}\label{rem:TM}
\begin{bmatrix}
\psi_{t+1}(\;\mathrm{is}^{\sharp}(\bs{e}_{br})\;) \\ \psi_{t+1}(\bs{e}_{br})
\end{bmatrix}
=\begin{bmatrix}1 & 0 \\ 0 & (-1)^{\tau(\bs{e}_{br})}\end{bmatrix}\;C \begin{bmatrix}
\psi_t(\;\mathrm{is}(\bs{e}_{br})\;) \\ \psi_t(\bar{\bs{e}}_{br})
\end{bmatrix}.
\end{equation}
In other words, 
for any island arc $\bs{e}_{is}\in A_{is}\subset A(\rho, \tau)$ and for any bridge arc $\bs{e}_{br}\in A_{br}\subset A(\rho,\tau)$, 
\begin{align}
(W\psi)(\bs{e}_{is}) &= a\;\psi(\rho^{-1}(\bs{e}_{is}))+b\;\psi(\mathrm{br}(\bs{e}_{is}))\\
(W\psi)(\bs{e}_{br}) &= (-1)^{\tau(\bs{e}_{br})}\left\{c\;\psi(\mathrm{is}(\bs{e}_{br}))+d\;\psi(\bar{\bs{e}}_{br})\right\}
\end{align}
for any $\psi\in \mathbb{C}^{A(\rho,\tau)}$. 
\end{definition}
\begin{figure}[hbtp]
    \centering
    \includegraphics[keepaspectratio, width=100mm]{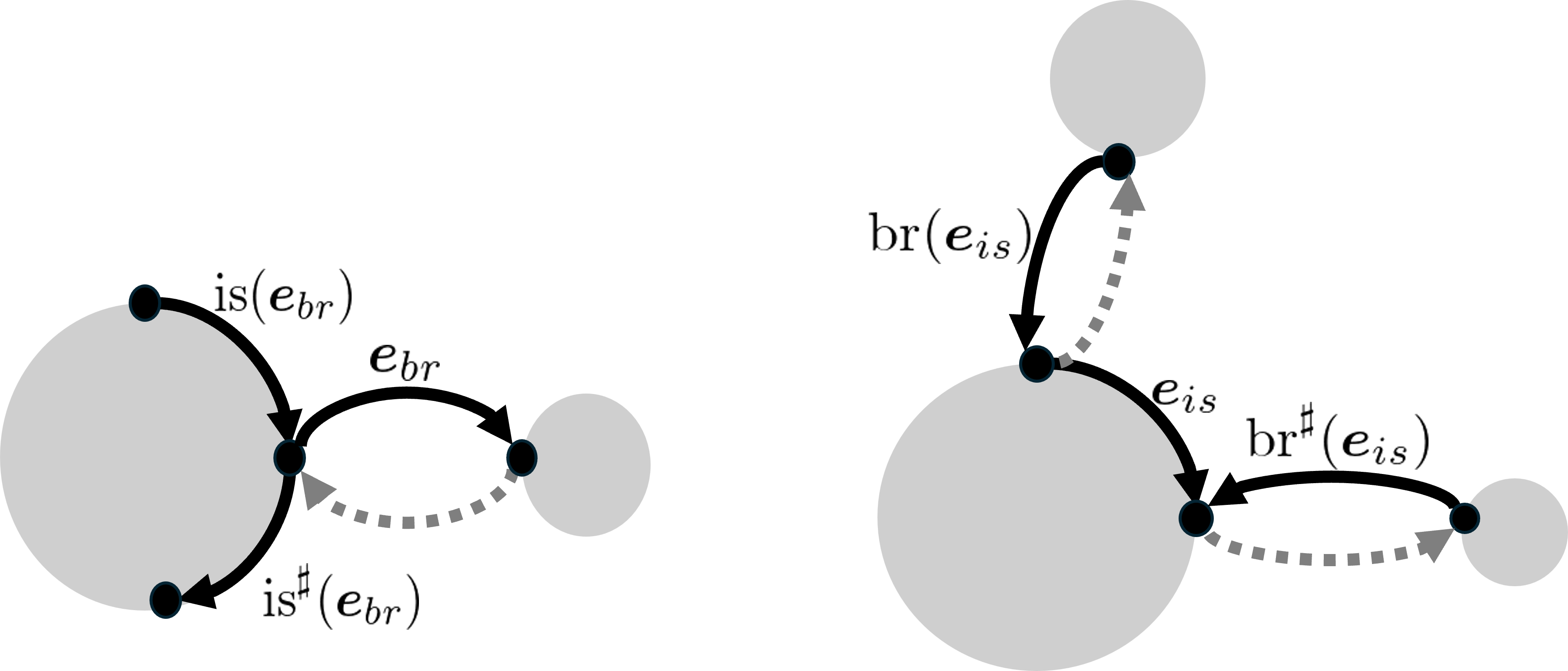}
    \caption{The definitions of $\mathrm{is}(\bs{e}_{br})$, $\mathrm{is}^\sharp(\bs{e}_{br})$, $\mathrm{br}(\bs{e}_{is})$ and $\mathrm{br}^\sharp(\bs{e}_{is})$.}
    \label{fig:isbr}
\end{figure}
\begin{figure}[hbtp]
    \centering
    \includegraphics[keepaspectratio, width=100mm]{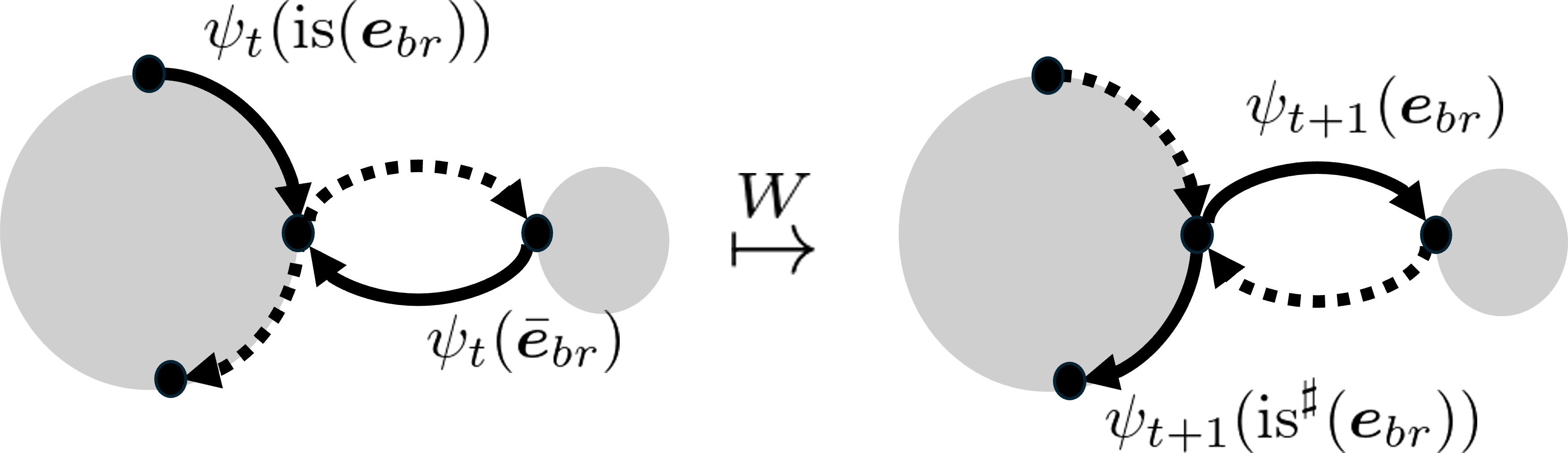}
    \caption{The local time evolution}
    \label{fig:LocTE}
\end{figure}

\subsection{Extension to an infinite system}\label{sect:TEIF}
For the blow-up graph $G(\rho,\tau)$, 
let us choose a subset of $A_{is}$ by $\delta A_{is}$ as the boundary to the ``outside". 
We deform this blow-up graph to an infinite graph by the following procedure. See Figure~\ref{fig:qypr}. 
\begin{definition}[The blow-up graph with tails]\label{def:BUT}
The blow-up graph with tail induced by $G(\rho,\tau)$ is given as follows. 
\begin{enumerate}
\item Each $\xi\in \delta A_{is}$ is divided by two arcs by replacing $\xi$ with two new arcs $\xi^-$ and $\xi^+$ satisfying with $o(\xi)=o(\xi^-)$, $t(\xi^-)=o(\xi^+)$ and $t(\xi^+)=t(\xi)$. 
\item Semi-infinite path $\mathbb{P}_e$ with the root vertex $o(\mathbb{P}_e)$ is joined by identifying $o(\mathbb{P}_e)$ with $t(\xi^-)=o(\xi^+)$. 
The pair of $(\xi^-,\xi^+)$ is called the {\it quay} of $e\in\delta A_{is}$. 
The sets of such as all $\xi^-$ and all $\xi^+$ are denoted by $\delta A_{qy}^-$ and $\delta A_{qy}^+$, respectively. 
\item Every $e\in A(\rho,\tau)\setminus \delta A_{is}$ is left as it is, without any transformations. 
\end{enumerate}
\end{definition}
\begin{figure}[hbtp]
    \centering
    \includegraphics[keepaspectratio, width=70mm]{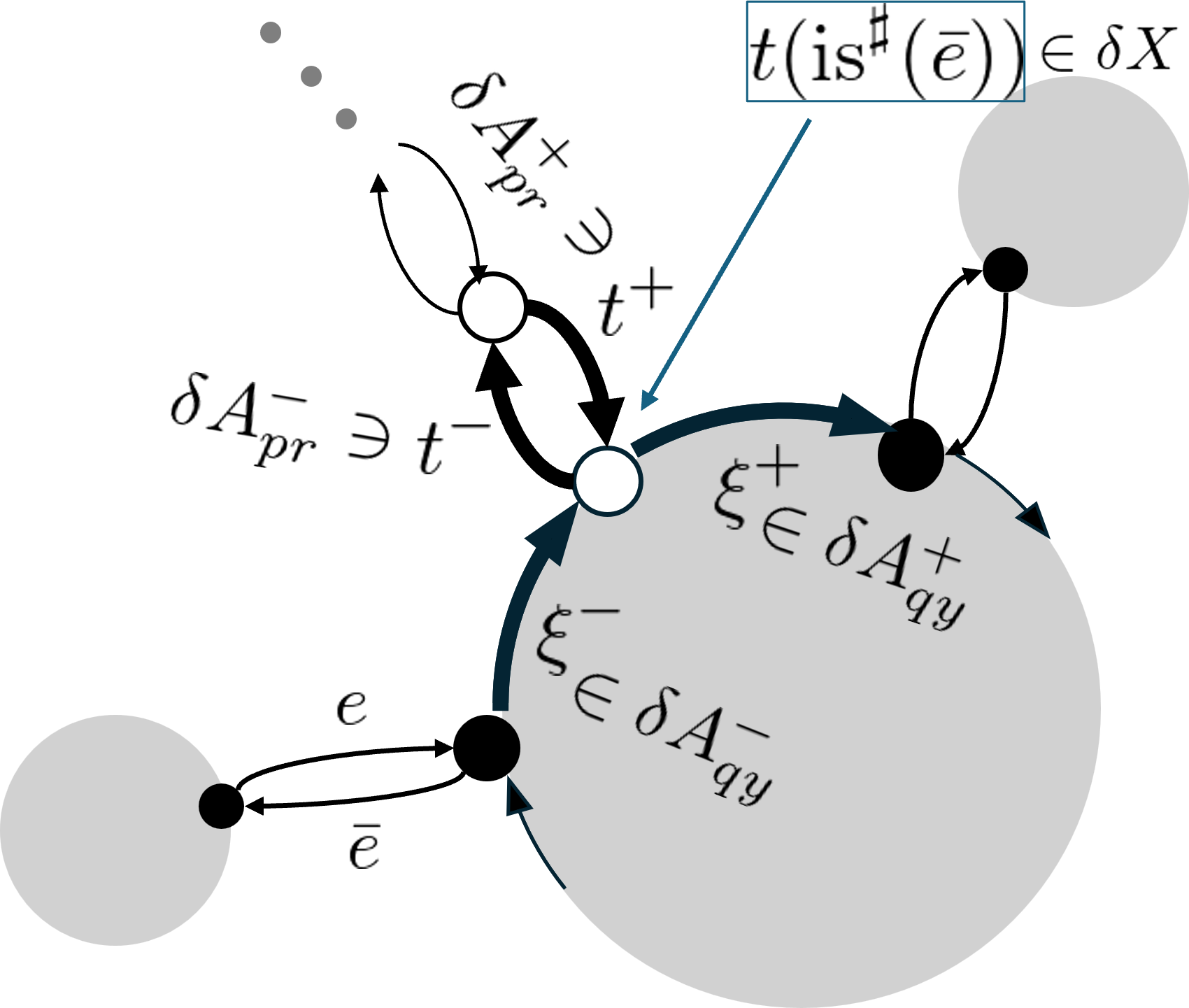}
    \caption{The quay and pier of the blow-up graph: The gray discs represent some islands of the blow-up graph. The semi-infinite path with the white vertices is the tail. The replaced arcs in the deformation process (1) with $\xi\in \delta A$ are called the quay arcs. The tail is inserted between the quay arcs $\xi^-$ and $\xi^+$. The arcs of of tails incident to the internal graph are called the pier arcs.   }
    \label{fig:qypr}
\end{figure}
The resulting graph is denoted by 
${\tilde{G}}(\rho,\tau)=(\tilde{X}(\rho,\tau),\tilde{A}(\rho,\tau))$, which is an infinite graph. 
See Figure~\ref{fig:FigBU} (d). 
The set of boundary vertices of $\tilde{G}(\rho,\tau)$ is defined by 
\[\delta X=\bigcup_{e\in\delta A_{is}}o(\mathbb{P}_e)\subset {\tilde{X}}(\rho,\tau) .\] 
The set of arcs of tails is denoted by $A_{tl}$. 
The subset of $A_{tl}$ called the {\it pier} is defined as  
\[ \delta A_{pr}=\delta A_{pr}^+ \sqcup \delta A_{pr}^-, \]
where 
\[ \delta A_{pr}^{+}= \{t^+\in A_{tl} \;|\; t(t^+)\in \delta X \},\;\delta A_{pr}^{-}= \{t^-\in A_{tl} \;|\; o(t^-)\in \delta X \}. \]
Note that for each island having a boundary vertex, the local rotation in $\tilde{G}(\rho,\tau)$ is naturally extended by the insertion of the tail.  We use the same symbol $\rho$ for the new rotation. 
Let us introduce the island $\tilde{A}_{is}$ and the bridge $\tilde{A}_{br}$ in $\tilde{G}(\rho,\tau)$ as follows. 
\begin{align*} 
\tilde{A}_{is} &= (A_{is}\setminus \delta A_{is})\sqcup \delta A_{qy}^{+} \sqcup \delta A_{qy}^{-}, \\
\tilde{A}_{br} &= A_{br} \sqcup \delta A_{pr}^+\sqcup \delta A_{pr}^{-}. 
\end{align*}
We assume $\tau (e)=0$ for any $e\in A_{tl}$. \\

The total state space here is extended to 
\[ \tilde{\mathcal{H}}_{total}=\mathbb{C}^{\tilde{A}(\rho,\tau)}. \]
The time evolution operator $\tilde{W}$ on $\mathbb{C}^{\tilde{A}(\rho,\tau)}$ is given as follows, which is essentially the same as Definition~\ref{def:TE1} except the boundaries and the tails, but the setting of the initial state is crucial to obtain the stationary state: 
\begin{definition}
Let $\Psi_t$ be the $t$-th iteration of $\tilde{W}$, such that 
$\Psi_{t}=\tilde{W}\Psi_{t-1}$ $t=1,2,\dots$
Here the time eovlution operator $\tilde{W}$ and the initial state $\Psi_0$ are defined as follows. 
\\

\noindent {\bf The time evolution operator $\tilde{W}$.}
\begin{enumerate}
\item For $\bs{e}\in A_{tl}\setminus \delta A_{pr}^-$:
The arcs of a tail are put by $e_0,e_1,\dots$ and $\bar{e}_0,\bar{e}_1,\dots$ with $t(e_0)\in \delta X$ $t(e_{j+1})=o(e_j)$ $j=0,1,2,\dots$. 
The time evolution on the tail is {\it free}, that is, 
\begin{equation}\label{eq:free} 
\Psi_{t+1}(e_{j})=\Psi_{t}(e_{j+1}),\;\Psi_{t+1}(\bar{e}_{j+1})=\Psi_{t}(\bar{e}_{j}) 
\end{equation}
for any $j=0,1,2,\dots$. 
\item For $\bs{e} \in \tilde{A}_{is}\cup A_{br}\cup \delta A_{pr}^-$: 
Let $a,b,c$ and $d$ are the same as in Definition~\ref{def:TE1}. 
For any island arc $\bs{e}_{is}\in \tilde{A}_{is}$, and for any ``bridge or pier" arc $\bs{e}_{br}\in A_{br}\cup \delta A_{pr}^-$, 
\begin{align*}
\Psi_{t+1}(\bs{e}_{is}) &= a\;\Psi_t(\rho^{-1}(\bs{e}_{is}))+b\;\Psi_t(\mathrm{br}(\bs{e}_{is})),\\
\Psi_{t+1}(\bs{e}_{br}) &= 
(-1)^{\tau(\bs{e}_{br})}\left\{c\;\Psi_t(\mathrm{is}(\bs{e}_{br}))+d\;\Psi_t(\bar{\bs{e}}_{br})\right\}.  
\end{align*}
Here if $\bs{e}_{is}\in\delta A_{qy}^{\pm}$, the bridge arc $\mathrm{br}(\bs{e}_{is})\in \delta A_{pr}^-$ is the pier from $t(\bs{e}_{is})$;
if $\bs{e}_{br}\in \delta A_{pr}^-$, the island arc $\mathrm{is}(\bs{e}_{br})\in \delta A_{qy}^{-}$ is the quay of the boundary vertex $o(\bs{e}_{br})$.
\end{enumerate}
\noindent{\bf The initial state $\Psi_0$.} \\
Let each tail be labeled by each element of  $\delta A_{pr}^+$. Prepare the sequence of complex values $(\alpha_{e})_{e\in \delta A_{pr}^+}$. 
Then we set the following uniformly bounded initial state on the tails:  
\begin{equation}\label{eq:initialstate} 
\Psi_0(e')=
\begin{cases}
\alpha_{e} & \text{: $e'\in A(\mathbb{P}_e)$,  $\dist (t(e'),o(\mathbb{P}_e))>\dist (t(e'),o(\mathbb{P}_e))$ ($e\in \delta A_{is}^+$),}\\
0 & \text{ otherwise.}
\end{cases} 
\end{equation}
\end{definition}
Note that the initial state $\Psi_0$ is bounded but no longer square summable. Such a setting of initial state  provides a constant inflow from the tails at every time step.
On the other hand, the setting of the time evolution on the tails implies that a walk never returns to the interior once it goes out to the outside, which can be regarded as an outflow from the interior. 
In this setting, the convergence to a fixed point of this quantum walk is ensured with the following meaning: 
\begin{proposition}[\cite{HS}]
For any $e \in \tilde{A}(\rho,\tau)$,
\[\exists\lim_{t\to\infty} \Psi_t(e)=\Psi_\infty(e). \]
\end{proposition}

\section{Unitary equivalence}
Let $(G,\rho,\tau)$ be a rotation system whose underlying abstract graph is $G=(X,A)$. 
Let us deform the rotation system as follows which is corresponding to Operation {\bf iii}.\\

\noindent{\bf Operation} ($\star$) \\
Choose a one vertex from $G$, say $x\in X$. Then reverse its rotation and switch all the edge types of the incident edges of $x$. 
Such a rotation system is denoted by $(G,\rho^{(x)},\tau^{(x)})$. \\

\noindent By Lemma~\ref{lem:switch}, the resulting embeddings on a closed surface of the rotation systems of the double covering graph, $(G^\tau,\rho\oplus\rho^{-1},\mathrm{id})$ and $(G^{\tau},\rho^{(x)}\oplus{\rho^{(x)}}^{-1},\mathrm{id})$, are equivalent embedding to each other. 
Under these equivalent two embeddings, we take the blow-up and choose the boundary island arcs $\delta A$ from $G(\rho,\tau)$ and its isomorphic boundary $\delta A'$ from $G(\rho^{(x)},\tau^{(x)})$.  
The resulting blow-up graphs are denoted by 
$\tilde{G}(\rho,\tau)$ and $\tilde{G}(\rho^{(x)},\tau^{(x)})$, respectively.
Note that they are still isomorphic to each other: $\tilde{G}(\rho,\tau)\cong \tilde{G}(\rho^{(x)},\tau^{(x)})$ under the following bijection $\phi$ for the incident arcs of islands $x$ and $x'$: \\

\noindent {\bf Bijection $\phi: \tilde{G}(\rho,\tau)\to \tilde{G}(\rho',\tau')$} satisfies the following: 
\begin{enumerate}
\item island : \\
$e_{is}=(e,\rho(e)) \in A_{is}((x,k))\subset A(\rho,\tau)$ \\ with $o(e)=(y,k')$ and $o(\rho(e))=(z,k'')$ in $G^\tau$ \\
$\leftrightarrow $\\
$\phi(e_{is})=(e',\rho(e')) \in A_{is}((x,k+1))\subset A(\rho,\tau)$ \\with $o(e')=(z, k''+1)$ and $o(\rho(e'))=(y, k'+1)$ in $G^\tau$. 
\item bridge : \\
$e\in A_{br}\subset A(\rho,\tau)$ \\
with $o(e)=(x,\ell)$ and $t(e)=(y,m)$ in $G^\tau$ \\ 
$\leftrightarrow$ \\
$\phi(e)\in A'_{br}\subset A(\rho',\tau')$ \\
with $o(e)=(x,\ell+1)$ and $t(e)=(y,m+1)$ in $G^\tau$
\end{enumerate}
The other arcs are left nothing as it is. \\

In this section, we show that the time evolution operators $\tilde{W}$ on $\tilde{G}(\rho,\tau)$ and $\tilde{W}^{(x)}$ on $\tilde{G}(\rho^{(x)},\tau^{(x)})$ are unitarily equivalent. Indeed we have the following proposition. 
\begin{proposition}\label{prop:uniequ}
Let the rotation system $(G,\rho,\tau)$ and $(G,\rho',\tau')$ are equivalent embeddings to each other, and let the boundaries of the induced blow-up graphs are also equivalent to each other. 
Then the time evolution operators $W$ on $\tilde{G}(\rho,\tau)$ and $W'$ on  $\tilde{G}(\rho',\tau')$ are unitarily equivalent, that is, 
there is an unitary map $\mathcal{U}$ such that 
\[ W'=\mathcal{U}^*\;W\;\mathcal{U}. \]
\end{proposition}
\begin{proof}
It is sufficient to consider the case for $\rho'=\rho^{(x)}$ for a fixed arbitrary vertex $x\in X$.
By Lemma~\ref{lem:switch}, ``the reversing the rotation of $x$" and ``the reversing all the incident edges of $x$" are reflected in the corresponding time evolution operator on $\tilde{G}(\rho,\tau)$:
\begin{enumerate}
    \item[{\bf Op(1)}] Switching the labels of all the arcs associated to $(x,0)$ and $(x,1)$ following the bijection $\phi$ ($\leftrightarrow$ Reversing the rotation of $x$);  
    \item[{\bf Op(2)}] Changing the twist $\tau$ to $\tau'$ by 
    \[ \tau'(e)=\begin{cases} 
    \tau(e)+1 & \text{: $e$ is incident to $(x,0)$ or $(x,1)$ in $G^\tau$,} \\ 
    \tau(e) & \text{: otherwise.}
    \end{cases} \]
    ($\leftrightarrow$ Reversing all the incident edges of $x$)
\end{enumerate}
Obviously, {\bf Op(1)} has the corresponding unitary map. This map is denoted by $\mathcal{U}_\phi$, that is, 
\[ (\mathcal{U}_{\phi}\psi)(e)=\psi(\phi^{-1}(e)). \] 

Then in the rest of the discussion, let us find the corresponding unitary map of {\bf Op(2)} following \cite[Lemma 2.4]{HiguchiShirai}. 
To this end, for the blow-up graph $\tilde{G}(\rho,\tau)$, set 
\[\mathcal{A}:=\{ \theta: A(\rho,\tau)\to \mathbb{C}\;|\; \theta(e)=0\text{ for all $e\in A_{is}$},\;\theta(\bar{e})=-\theta(e)\text{ for all $e\in A_{br}$} \}.\] 
Note that if we set 
\begin{equation}\label{eq:theta1}
\theta_1(e)=\begin{cases} \pi & \text{: $e\in A_{br}$ with $\tau(e)=1$,} \\ 0 & \text{: otherwise,} \end{cases} 
\end{equation}
then $\theta_1\in \mathcal{A}$ and 
$e^{\im \; \theta_1(e)}=(-1)^{\tau(e)}$ for any $e\in A_{br}$, where $\im =\sqrt{-1}$. 
On the other hand if we set 
\begin{equation}\label{eq:theta2}
\theta_2(e)=\begin{cases} \pi & \text{: $e\in A_{br}$ with $\tau'(e)=1$,} \\ 0 & \text{: otherwise,} \end{cases} 
\end{equation}
then $\theta_2\in \mathcal{A}$ and $e^{\im \; \theta_2(e)}=(-1)^{\tau'(e)}$ for any $e\in A_{br}$.
For $\theta\in \mathcal{A}$, and an arbitrary path in $G(\rho,\tau)$, $p=(a_1,a_2,\dots,a_k)$ with $t(a_j)=o(a_{j+1})$ $(j=1,\dots,k-1)$, we define 
\[  \int_p \theta := \sum_{j=1}^k \theta(e_j).  \]
Set $\tilde{\rho}:=\rho\oplus \rho^{-1}$.
A path $p=(a_1,\dots,a_s)$ with $a_1\in A_{br}$ and $a_s\in A_{is}$ in $G(\rho,\tau)$ can be represented by 
\[  p=(e_1,\Gamma_1,e_2,\Gamma_2,\dots,e_r,\Gamma_r), \]
where $e_1,\dots,e_r\in A_{br}$,  
\[ \Gamma_j=(\tilde{\rho}^0(\xi_j),\tilde{\rho}^1(\xi_j),\dots,\tilde{\rho}^{k_j}(\xi_j)),\;\;(\tilde{\rho}^{\ell}(\xi_m)\in A_{is}) \]
with $e_1=a_1$ and $\rho^{k_r}(\xi_r)=a_s$. 
Here if $\Gamma_j=\emptyset$, then $e_{j+1}=\bar{e}_j$. 
The inverse path of $p$ is defined by 
\[ \bar{p}=(\bar{\Gamma}_r,\bar{e}_r,\dots, \bar{\Gamma}_2,\bar{e}_2,\bar{\Gamma}_1,\bar{e}_1), \]
where 
\[ \bar{\Gamma}_j=(\tilde{\rho}^{k_j+1}(\xi_j),\tilde{\rho}^{k_j+2}(\xi_j),\dots,\tilde{\rho}^{d_j-1}(\xi_j)). \]
The inverse path of $p$ with $a_1\in A_{is}$ and $a_s\in A_{is}$ and so on is also defined in the same way. 
\begin{lemma}\label{lem:welldefine}
Set $\alpha, \beta\in \mathcal{A}$.  
If $\int_c \alpha =\int_c \beta $
for any closed path $c$ in $G(\rho,\tau)$, then 
$\int_p (\beta-\alpha)=\int_{p'} (\beta-\alpha)$
for any path $p$ and $p'$ with $o(p)=o(p')$ and $t(p)=t(p')$. 
\end{lemma}
\begin{proof}
It holds that $\int_{\bar{p}} \theta=-\int_{p} \theta$ for any path $p$ by the definition. 
Note that $p\cup \bar{p'}$ is a closed path. Then we have 
\[ 0=\int_{p \cup \bar{p}}(\beta-\alpha) 
=\int_{p}(\beta-\alpha)+\int_{\bar{p'}} (\beta-\alpha)
=\int_{p}(\beta-\alpha)-\int_{p'} (\beta-\alpha), \]
which is the conclusion.
\end{proof}
Lemma~\ref{lem:welldefine} tells us that the value $\int_p (\theta_2-\theta_1)$ is independent of choice of path from $o(p)$ to $t(p)$. 
Let us fix a vertex $x_*$. Since $\int_p (\theta_2-\theta_1)$ for a path $p$ starting from $x_*$ is determined by $t(p)=x$, we denote such a value by $\Delta(x)$, which is well-defined.  

Let $\psi_t$ (resp. $\psi_t'$) be the $t$-th iteration of $\tilde{W}$ (resp. $\tilde{W'}$) with the twist $\tau$ (resp. $\tau'$) such that $\psi_{t+1}=\tilde{W}\psi_t$ (resp. $\psi'_{t+1}=\tilde{W}'\psi_t'$). 
By Remark~\ref{rem:TM}, we have 
\begin{equation}\label{eq:TM2}
\begin{bmatrix}
\psi_{t+1}(\xi) \\ \psi_{t+1}(\epsilon)
\end{bmatrix}
=\begin{bmatrix}e^{\im\;\theta_1(\xi)} & 0 \\ 0 & e^{\im\; \theta_1(\epsilon)}\end{bmatrix}\;C \begin{bmatrix}
\psi_t(\xi^\flat) \\ \psi_t(\bar{\epsilon})
\end{bmatrix},  
\end{equation}
for any $\epsilon\in A_{br}$,  
where $\xi^\flat=\mathrm{is}(\epsilon)$ and $\xi=\mathrm{is}^\sharp(\epsilon)$.
Here $\theta_1$ is defined in (\ref{eq:theta1}).
Now let us set $\theta_2$ by (\ref{eq:theta2}). 
{\bf Operation $(\star)$} keeps the parity of any paths passing through the vertex $x$, because $\tau(e)+\tau(e')$ with $t(e)=x=o(e')$ in $G$ is changed to \[ \tau'(e)+\tau'(e')=\tau(e)+1+\tau(e')+1=\tau(e)+\tau(e'), \]
which implies $\int_c \theta_1=\int_{c} \theta_2$ for any cycle $c$ in $\tilde{G}$. 
Let us consider a path $p$ with $o(p)=x_*$ and $t(p)=o(\xi)=o(\epsilon)$ in $\tilde{G}$. Thus we have 
\begin{align*} 
\Delta(t(\xi))&=\int_{p+\xi} (\theta_2-\theta_1) 
= \int_{p} (\theta_2-\theta_1) + \theta_2(\xi)-\theta_1(\xi) \\
&= \Delta(o(\xi))+ \theta_2(\xi)-\theta_1(\xi),\\
\Delta(t(e))&=\int_{p+ e} (\theta_2-\theta_1) 
= \int_{p} (\theta_2-\theta_1) + \theta_2(e)-\theta_1(e) \\
&= \Delta(o(\epsilon))+ \theta_2(e)-\theta_1(e),\\
\end{align*}
Here for a path $p$ with $t(p)=x$ and the fixed vertex $x_*$, 
\begin{equation}\label{eq:Delta}
\Delta(x)=\int_{p} \theta_2-\theta_1  
\end{equation} 
is well-defined by Lemma~\ref{lem:welldefine}. 
Inserting them into (\ref{eq:TM2}), we  have 
\begin{align*} \begin{bmatrix} e^{\im\;\Delta(\;t(\xi)\;)}\;\psi_{t+1}(\xi) \\ e^{\im\;\Delta(\;t(\epsilon)\;)}\; \psi_{t+1}(\epsilon) \end{bmatrix}
& = \begin{bmatrix} e^{\im\;\theta_2(\xi)} & 0 \\ 0 & e^{\im\;\theta_2(\epsilon) } \end{bmatrix} \begin{bmatrix} e^{\im\;\Delta(\;o(\xi)\;)} & 0 \\ 0 & e^{\im\;\Delta(\;o(\epsilon)\;)} \end{bmatrix} \;C\; \begin{bmatrix} \psi_t(\xi^\flat) \\ \psi_t(\bar{\epsilon}) \end{bmatrix} \\
&= \begin{bmatrix} e^{\im\;\theta_2(\xi)} & 0 \\ 0 & e^{\im\;\theta_2(\epsilon) } \end{bmatrix}\;C\;
\begin{bmatrix} e^{\im\;\Delta(\;t(\xi^\flat)\;)}\;\psi_t(\xi^\flat) \\ e^{\im\;\Delta(\;t(\bar{\epsilon})\;)}\;\psi_t(\bar{\epsilon}) \end{bmatrix}
\end{align*}
The second equality derives from $o(\xi)=o(\epsilon)$ which gives the commutativity of the second diagonal matrix and $C$ in the first equality. 
Therefore introducing the unitary map $\mathcal{U}_\Delta$ by
\[ (\mathcal{U}_\Delta\psi)(a)=e^{\im\;\Delta(t(a))}\;\psi(a), \]
we obtain 
\[ W'=(\mathcal{U}_\Delta\mathcal{U}_\phi)^*\; W\;(\mathcal{U}_\Delta\mathcal{U}_\phi).  \]
\end{proof}
\begin{remark}
For a vertex $(x,j)\in X(\rho,\tau)$, define $\mathrm{sheet}(x,j):=j\in \mathbb{Z}_2$. 
Then the unitary map $\mathcal{U}=\mathcal{U}_\Delta\;\mathcal{U}_{\phi}$ is described by 
\[ (\mathcal{U}\psi)(a)=
\begin{cases}
\psi(\phi^{-1}(a)) & \text{:  $\mathrm{sheet}(t(a))+\mathrm{sheet}(x_*)=\mathrm{sheet}(\phi(t(a)))+\mathrm{sheet}(\phi(x_*))$,}\\
-\psi(\phi^{-1}(a)) & \text{:  $\mathrm{sheet}(t(a))+\mathrm{sheet}(x_*) \neq \mathrm{sheet}(\phi(t(a)))+\mathrm{sheet}(\phi(x_*))$.}
\end{cases} \]
\end{remark}
\section{Scattering}\label{sect:scat}
\subsection{The scattering matrix represented by faces}
For the blow-up graph $G(\rho,\tau)$, 
let $\delta X=\{x_1,\dots,x_{\kappa}\}$ and $\delta A_{pr}^+=\{e_1,\dots,e_{\kappa}\}$ such that $t(e_j)=x_j$. 
The stationary state is denoted by $\Psi_\infty\in\mathbb{C}^{\tilde{A}(\rho,\tau)}$. 
Let us represent the inflow from the outside by $\bs{\alpha}_{in}\in \mathbb{C}^{\delta X}$ such that 
\[ \bs{\alpha}_{in}(x)=\Psi_\infty(e) \text{ with $t(e)=x$} \]
for any $x\in \delta X$, $\ell\in \mathbb{Z}_2$, 
while the outflow to the outside by  $\bs{\beta}_{out}\in \mathbb{C}^{\delta X}$ such that 
\[ \bs{\beta}_{out}(x)=\Psi_\infty(\bar{e}) \text{ with $t(e)=x$} \] 
for any $x\in \delta X$. 
The scattering matrix $S:\mathbb{C}^{\delta X}\to \mathbb{C}^{\delta X}$ is defined by 
\begin{equation}\label{eq:defscat}  
\bs{\beta}_{out}=S\bs{\alpha}_{in}, 
\end{equation}
which is determined by the rotation system $G(\rho,\tau)$ and $\delta X$, and independent of $\bs{\alpha}_{in}$, $\bs{\beta}_{out}$. 
It is shown in \cite{HS} that such a matrix $S$ exists and is unitary, but its explicit expression is up to the individual setting. Here we characterize the scattering matrix $S$ by the geometric information. \\

We prepare important graph notions which express the scattering matrix. 
Let $f=(e_0,\dots,e_{\kappa-1})\in F$ be a facial closed walk of the rotation system $(G^\tau,\rho\oplus\rho^{-1},\mathrm{id})$ with length $\kappa=|f|$. The extended facial walk of $\tilde{f}$ in $A(\rho,\tau)$ is defined by just alternatively inserted corresponding island arcs into each arc of $f$; (we use the same notation for the extended facial walk by $f$): 
\begin{equation}\label{eq:externalwalk}
    f:=(\;\xi_0,e_0,\dots,\xi_{\kappa-1},e_{\kappa-1}\;), 
\end{equation} 
where $e_j\in A_{br}\cong A$ and $\xi_j\in A_{is}$ with 
\begin{equation}\label{eq:condition}
\xi_j = \mathrm{is}(e_j), \;
e_{j}= \mathrm{br}(\xi_{j+1})
\end{equation}
for any $j\in \{0,1,\dots,\kappa-1\}$ in the modulus of $\kappa$. 
We remark that $\xi_j$ represents the sequence of two quay arcs $\xi_j^{-}$ and $\xi_j^{+}$ if $\xi_j\in \delta A_{is}$. 
Let $\delta F\subset F$ be the set of (extended) facial walk passing through a boundary vertex $\delta X$ in $\tilde{G}$. 
For any facial walk $f$, there exists a chiral facial walk $f^{*}$ defined by going around the opposite direction on the opposite sheet.  
We introduce the following useful lemma to consider the scattering.
\begin{lemma}\label{lem:key}
Assume $d\in\mathbb{R}$. 
Let $\psi_\infty\in \mathcal{H}$ be the stationary state. Set $\omega=-\det (C)$. 
For each facial walk $f\in \delta F$ of the orientation system $(G,\rho,\tau)$ represented by a sequence of $\Omega$,  
\begin{equation*}
    f:=(\;\xi_0,e_0,\xi_1,e_1,\dots,\xi_{\kappa-1},e_{\kappa-1}\;), 
\end{equation*} 
 we have 
\begin{equation}\label{eq:key}
\psi_\infty(\xi'_{j+1})=
(-1)^{\tau(e_j)}\omega \;\psi_\infty(\xi_j''),\;(j=0,1,\dots,\kappa-1) 
\end{equation}
where if $\xi_{j+1}\in \delta A_{is}$, then $\xi_{j+1}'=\xi_{j+1}^{-}$, otherwise $\xi_{j+1}$; if $\xi_j\in \delta A_{is}$, then $\xi_{j}''=\xi_{j}^+$, otherwise $\xi_j''=\xi_j$. 
\end{lemma}
\begin{proof}
For any bridge arcs $e\in A_{br}$, (of course) the vertices $o(e)$ and $t(e)$ are connected by the bridge arcs $e^R_1:=e$ and $e^L_1:=\bar{e}$, moreover the island arcs $e_0^R:=\mathrm{is}(e)$ and 
$e_0^L:=\mathrm{is}^\sharp(e)$ are connected to $o(e)$, while the island arcs $e_2^R:=\mathrm{is}^\sharp(\bar{e})$ and $e_2^L:=\mathrm{is}(\bar{e})$ are connected to $o(t)$. 
Let us set $\psi:=\psi_\infty$. 
The local scatterings at $o(e)$ and $t(e)$ imply 
\begin{align}\label{eq:trandfermat}
\begin{bmatrix}
-a & 1 \\ -c' & 0 
\end{bmatrix}
\begin{bmatrix}
\psi(e_0^R) \\ \psi(e_0^L)
\end{bmatrix}
=
\begin{bmatrix}
0 & b \\ -1 & d' 
\end{bmatrix}
\begin{bmatrix}
\psi(e_1^R) \\ \psi(e_1^L)
\end{bmatrix}\text{ and }
\begin{bmatrix}
1 & -a \\ 0 & -c' 
\end{bmatrix}
\begin{bmatrix}
\psi(e_2^R) \\ \psi(e_2^L)
\end{bmatrix}
=
\begin{bmatrix}
b & 0 \\ d' & -1 
\end{bmatrix}
\begin{bmatrix}
\psi(e_1^R) \\ \psi(e_1^L)
\end{bmatrix}
\end{align}
where $c'=(-1)^{\tau(e)}c$, $d'=(-1)^{\tau(e)}d$. 
Therefore we have  
\begin{align*}
\begin{bmatrix}
\psi(e_2^R) \\ \psi(e_2^L)
\end{bmatrix}
& = \begin{bmatrix}
1 & -a \\ 0 & -c' 
\end{bmatrix}^{-1} 
\begin{bmatrix}
b & 0 \\ d' & -1 
\end{bmatrix}
\begin{bmatrix}
0 & b \\ -1 & d' 
\end{bmatrix}^{-1}
\begin{bmatrix}
-a & 1 \\ -c' & 0 
\end{bmatrix}
\begin{bmatrix}
\psi(e_0^R) \\ \psi(e_0^L)
\end{bmatrix} \\
&= \frac{1}{|c|^2\omega'} 
\begin{bmatrix} 
\omega' & 0 \\ 0 & 1
\end{bmatrix}
\begin{bmatrix}
1-{\bar{d'}}^2 & d'-\bar{d'} \\
-d'+\bar{d'} & 1-{d'}^2
\end{bmatrix}
\begin{bmatrix} 
\omega' & 0 \\ 0 & 1
\end{bmatrix}
\begin{bmatrix}
\psi(e_0^R) \\ \psi(e_0^L)
\end{bmatrix}
 \end{align*}
where $\omega'=(-1)^{\tau(e)}\omega$. 
Here the second equality is obtained by the properties of the unitarity matrix $C$, for examples, 
\begin{align}
&a=-\omega'\bar{d'},\;b=\omega'\bar{c'},\notag \\
&|a|^2+|b|^2=|d'|^2+|c'|^2=1 \label{eq:unitarity}
\end{align}
and so on. 
Then $d\in \mathbb{R}$ if and only if $\psi(e_2^R)=\omega'\psi(e_0^R)$ and $\psi(e_0^L)=\omega'\psi(e_2^L)$, 
which is the desired conclusion. 
\end{proof}
The local path structures around the bridges in the blow-up graph  lead  transfer matrices discussed in the spectral analysis on the discrete-time quantum walks on the one-dimensional lattice and gives the conclusion. 

For a facial walk $f=(\xi_0,e_0,\dots,\xi_{\kappa-1},e_{\kappa-1})\in \delta F$, set $f\cap \delta A_{is}=\{\xi_{k_0},\dots,\xi_{k_{q-1}}\}$ and $f\cap \delta X:=\{ x_{k_0},\dots,x_{k_{q-1}} \}$, where $x_{k_j}=t(\xi_{k_j}^{-})$ with $k_0<k_1<\cdots<k_{q-1}$ in the modulus of $q$. 
Here $1\leq q\leq \kappa$. 
(See Figure~\ref{fig:qypr}. )
Note that the $q$ tails interfere with the facial walk $f$. 
Put 
\begin{equation}\label{eq:choucho}
\bowtie_f(\ell,m):=\int_{x_{k_{\ell}}}^{x_{k_{m}}}\tau  \text{ and }\  \mathrm{dist}_f(\ell,m):=k_{\ell}-k_m,
\end{equation}
which are the parity of number of type-$1$ edge between $x_{m}$ and $x_{\ell}$, and the distance between the boundary vertices $k_m$ and $k_{\ell}$ along the facial walk $f$, respectively. 
To simplify the notation, let us put $f\cap \delta X=\{0,1,\dots,q-1\}$. 
We set the identity matrix, the weighted cyclic permutation matrix on $\mathbb{C}^{\{0,\dots,q-1\}}$ by 
\begin{align}\label{eq:op_ex_walk} 
& (I_{f}h)(j)
=h(j), \notag \\
& (P_{f}(\omega)h)(j)=(-1)^{\bowtie_f(j,j-1)} \omega^{\mathrm{dist}_f(j,j-1)} h(j-1)
\end{align}
for any $h\in \mathbb{C}^{\{0,\dots,q-1\}}$ and $j\in \{0,\dots,q-1\}$ in the modulus of $|f|$, respectively.  
Now we are ready to give the theorem for the scattering:
\begin{theorem}\label{thm:scattering0}
Assume $d\in \mathbb{R}$ and set $\omega=-\det C$. 
The scattering matrix is decomposed into the following $|F|$ unitary matrices as follows:  \[S=\bigoplus_{f\in F} S_{f}, \]
where 
\[ S_{f}=bc P_{f}(\omega)\;(I_{f}-aP_{f}(\omega))^{-1}+dI_{f}. \]
Here the operators induced by each external facial closed walk $I_f$ and $P_f(\omega)$ are defined in (\ref{eq:op_ex_walk}).
\end{theorem}
\begin{remark}
Since $S_f$ is the $q\times q$ unitary matrix and
$P_f^q(\omega)=I_f$, where $q=|f\cap \delta X |$,  then $S_f$ is rewritten by 
\[ S_f=\frac{bc}{1-a^{q} \omega^{|f|} }P_f(\omega)\left(I_f+(aP_f(\omega))+(aP_f(\omega))^2+\cdots+(aP_f(\omega))^{q-1}\right) .   \]
From this expression, if $f\cap \delta X=\{0,\dots,q-1 \}$, we have 
\[ (S_f)_{j,i}=\begin{cases} 
\displaystyle bc\;\frac{a^{j-i-1}\omega^{\dist_f(j,i)}}{1-a^{|f\cap \delta X|}} \times (-1)^{\bowtie_f(j,i)} & \text{: $i\neq j$,}\\
\\
\displaystyle bc\;\frac{a^{|f\cap \delta X|-1}\omega^{|f|}}{1-a^{|f\cap \delta X|}} +d & \text{: $i= j$.}
 \end{cases} \]
Here ``$j-i-1$" corresponds to the number of boundary vertices where the facial walk $f$ passes through from $k_i$ to $k_j$. 
\end{remark}
\begin{proof}
Pick up a facial walk $f=(\xi_0,e_0,\dots,\xi_{\kappa-1},e_{\kappa-1})\in \delta F$, and recall that $f\cap \delta X=\{ x_{j_0},\dots,x_{j_{q-1}} \}$ and $f\cap \delta A_{is}=\{\xi_{k_0},\dots,\xi_{k_q}\}$. 
Here $1\leq q\leq \kappa$. 
Note that the $q$ tails interfere with the facial walk $f$. The tail incident to the boundary $x_{k_i}$ is denoted by $\mathrm{Tail}_i$. 
Let us give the inflow $1$ from the $\mathrm{Tail}_i$, and consider the outflow to the tail $\mathrm{Tail}_j$ ($i,j\in\{0,1,\dots,q-1\}$).  
Let us see that 
the $(j,i)$ element of the scattering matrix $S$ can be calculated according to the number of ``nights" that a quantum walk stays at the face $f$ with the boundary vertices $x_{k_i}$ and $x_{k_j}$ from the time the quantum walk enters through entrance $i$ to the time the quantum walk leaves through exit $j$.

Consider the case for $i=0$. 
Let $t^{in}_j\in \delta A_{pr}^+$ and $t^{out}_j\in \delta A_{pr}^-$ be the arcs of $\mathrm{Tail}_j$ incident to the vertex $x_{k_j}\in \delta X$ (see Figure~\ref{fig:qypr}).  
Let ``the day trip walk" from $\mathrm{Tail}_0$ to $\mathrm{Tail}_j$ be defined by 
\[ (t_0^{in},\bs{w}_0,t_j^{out} ), \]
where \[ \bs{w}_0:=(\xi_{k_0}^+,e_{k_0},\dots, \xi_{k_1}^-,\xi_{k_1}^+,e_{k_1},\dots,
\xi_{k_{j-1}}^-,\xi_{k_{j-1}}^+,e_{k_{j-1}},\dots, \xi_{k_j}^-), \]
and the ``$r$-night walk" from   $\mathrm{Tail}_0$ to $\mathrm{Tail}_j$ be defined by 
\begin{equation*}
(t_0^{+},\bs{w}_0,
\overbrace{\xi_{k_j}^+,e_{k_j},\dots, \xi_{k_{j+1}}^-,\xi_{k_{j+1}}^+,e_{k_{j+1}},\dots,
\xi_{k_{j-1}}^-,\xi_{k_{j-1}}^+,e_{k_{j-1}},\dots, \xi_{k_j}^-,}^{\text{$r$ times (nights)}}
t_j^{-})
\end{equation*}
By Lemma~\ref{lem:key}, the weight associated with the moving along the facial walk from $\xi_\ell\in A_{is}$ to $\xi_{\ell+1}\in A_{is}$ must be $\omega \times (-1)^{\tau(e_\ell)}$ in the stationary state. 
By the local time evolution denoted by the $2\times 2$-unitary matrix 
\[C=\begin{bmatrix} a & b \\ c & d \end{bmatrix},\] the weights associated with moving from the tail $t_0^{+}\in \delta A_{pr}^+$ to the quay $\xi_{k_0}^+\in A_{is}$, and from the quay $\xi_{k_j}^-\in A_{is}$ to the tail $t_j^{-}\in \delta A_{pr}^-$ are $c$ and $b$, respectively; the weight associated with moving from $\xi_{k_s}^{-}\in A_{is}$ to $\xi_{k_s}^{+}\in A_{is}$ is $a$. 
Remark that the weight of the closed path starting from $\xi_{{k_0}}^+$ and returning back to the same arc $\xi_{k_0}^{+}$ along the boundary face $f$ is $a^q\omega^\kappa$ since $(-1)^{\int_{f} \tau}=1$. 
Then set $\bs{s}^{(r)}(\cdot,0)\in \mathbb{C}^{\{0,1,\dots,q-1\}}$ ($r=0,1,\dots$) as 
the weight of $r$-night walk in the stationary state by $\bs{s}^{(r)}(j,0)=(a^q \omega^\kappa)^r\times  \bs{s}^{(0)}(j,0)$, $(r=1,\dots,\kappa-1)$ with
\[\bs{s}^{(0)}(j,0)
=\begin{cases}
b\;c\;a^{j-1}\;\omega^{\dist_f(j,0)}\;(-1)^{\bowtie_f(j,0)} & \text{: $j\neq 0$}\\
b\;c\;a^{q-1}\;\omega^{\kappa}\; & \text{: $j= 0$}
\end{cases} 
\]

The outflow from $j$ is obtained by the superposition of $\bs{s}^{(r)}(j,0)$'s $(r=0,1,2,\dots)$ because of the constant inflow $1$ at every time step from the tail $t_{k_0}$. Then we have 
\begin{align*}
(S)_{k_j,k_0} &= \sum_{r=0}^\infty \bs{s}^{(r)}(j,0)=
\begin{cases}
bc\frac{a^{j-1}\omega^{k_j-k_0}}{1-a^q\omega^\kappa}\times (-1)^{\bowtie_f(j,0)} & \text{: $j\neq 0$}\\
\\
bc\frac{a^{q-1}\omega^\kappa}{1-a^q\omega^\kappa}+d & \text{: $j= 0$}
\end{cases} \\
&= \frac{bc}{1-a^q\omega^\kappa} \left(P_f(\omega)+a P_f(\omega)^2+\cdots+a^{q-2}P_f^{q-1}(\omega)+a^{q-1}\omega^\kappa I_f\right)_{j,0}-d\delta_{j,0} \\
&= (S_f)_{k_j,k_0}.
\end{align*}
From the symmetricity of the rotation, we have 
\begin{equation}
(S)_{k_j,k_i} = \begin{cases} 
\displaystyle
 bc\; \frac{a^{j-i-1} \omega^{\dist_f(j,i)}}{1-a^q\omega^\kappa}\times (-1)^{\bowtie_f(j,i)} & \text{: $i\neq j$,}\\
\\
 \displaystyle bc\; \frac{a^{q-1}\omega^\kappa}{1-a^q\omega^\kappa}+d & \text{: $i= j$,}
\end{cases}
\end{equation}
which leads the conclusion. 
\end{proof}
Let $(G,\rho_1,\tau_1)$ and $(G,\rho_2,\tau_2)$ be the rotation systems which are equivalent to each other. Let $S_1$ and $S_2$ be the scattering matrices of those resulting embbedings, respectively. 
\begin{proposition}\label{prop:scatiso}
Let $S_1$ and $S_2$ be the above. 
Recall that $\Delta(x)$ is defined in (\ref{eq:Delta}).
Then we have 
\[ S_2=D^*S_1D, \]
where $D=\oplus_{f\in F}D_f$ is the diagonal matrix such that
\[ D_f=\mathrm{diag}[ e^{-i\Delta(t(e_0))},e^{-i\Delta(t(e_1))},\dots,e^{-i\Delta(t(e_{\kappa-1}))} ] \]
for a face $f=(e_0,e_1,\dots,e_{\kappa-1})$. 
\end{proposition}
\begin{proof}
Proposition~\ref{prop:uniequ} immediately leads to the conclusion. 
\end{proof}

Let us consider the case for the following special assignment of the tail, which can be constructed independently of the embedding. \\

\noindent {\bf The hedgehog tail assignment}:
We call the hedgehog tail assignment if   
\begin{equation}\label{eq:def:hedgehog}
\delta A_{is}=A_{is},
\end{equation} 
which is the setting of the tails so that a tail is inserted between each vertex in the islands. \\

In the hedgehog tail assignment, $\delta F=F$ holds. 
Pick up a facial walk $f=(\xi_0,e_0,\dots,\xi_{\kappa-1},e_{\kappa-1})\in \delta F=F$. 
Since  $\mathrm{dist}_f=1$, $|\delta X \cap f|=|f|$, 
$P_f(\omega)=\omega P_f$, where $P_f:=P_f(1)$ such that 
\begin{equation}\label{eq:hedgehogP}
(P_fh)(j)=(-1)^{\tau(e_{j-1})}h(j-1)
\end{equation}
for any $j=0,1,\dots,|f|-1$ and $h\in \mathbb{C}^{\{x_0,\dots,x_{|f|-1}\}}$.
Then Theorem~\ref{thm:scattering0} implies Theorem~\ref{thm:scattering} which is the special case for the hedgehog. 

\subsection{Detection of the orientability}
Let us estimate whether the underlying closed surface is orientable or not by observing the outflow of the internal graph to an inflow under the Assumption~\ref{ass:hd}. 

By the hedgehog assignment of the tails in (\ref{eq:def:hedgehog}), $\delta X$ is isomorphic to $A_{br}$
by the following one-to-one correspondance $\varphi$ between $A_{br}$ and $\delta X$: 
\begin{equation}\label{eq:bijection}
\varphi(e) = t(\mathrm{is}^\sharp(\bar{e}))
\end{equation}
for any $e\in A_{br}$ (see Figure~\ref{fig:qypr}).

Here  the definitions of $\xi^-$ and $\xi^+$ are defined in Definition \ref{def:BUT}.
We label the scattering matrix by $e\in A_{br}$ by using the bijection map $\varphi$ from now on.  
Let us introduce the following subset of $A_{br}$. For any islands of $x\in X^\tau$,  we define 
\[ \delta A(x):=\{e\in A_{br} \;|\; t(e)\in t(A_{is}(x))\}, \]
which corresponds to the set of the vertices having a tail in the island $x$. 

For example, let us insert $(+1)$-inflow into a single tail. 
The following theorem tells us that once we find an island emitting outflows whose signatures are not invariant, the underlying closed surface must be unorientable. 
\begin{theorem}\label{thm:detection}
For any islands of $x,y\in X^\tau$ with $x\neq y$, 
let the submatrix of $S$ formed from rows $\delta A(x)$ and columns $\delta A(y)$ be denoted by 
\[ S|_{x,y}:=[(S)_{\varphi(e),\varphi(e')}]_{e\in \delta A(x),\; e'\in \delta A(y)}. \]
Then 
under Assumption~\ref{ass:hd} with $a>0$,  
the underlying surface is orientable if and only if the signature of all the elements of 
$(S|_{x,y})_{\varphi(e),\varphi(e')}$'s with $(S|_{x,y})_{\varphi(e),\varphi(e')}\neq 0$ are invariant.

\end{theorem}
\begin{proof}
Since $a>0$, 
\[ (S)_{\varphi(e),\varphi(e')}/|S_{\varphi(e),\varphi(e')}|=
 \begin{cases}(-1)^{\bowtie_f(t(e),t(e'))}
& \text{: $e$ and $e'$ are included in the same face, }\\
0 & \text{: otherwise.}
\end{cases}
\]
Here 
$ \bowtie_f (z,z')
=\int_{z}^{z'}\tau=\sum_{e\in f}\tau(e)=\sum_{e\in f\cap A_{br}}\tau(e)$ because $\tau(\xi)=0$ for any island arcs $\xi$. 
We remark that the resulting embedding of the rotation system $(G,\rho,\tau)$ is orientable if and only if $\int_c \tau=0$ for any closed paths in the rotation system $(G,\rho,\tau)$ by {\bf Operation} in Section~\ref{sect.RS}. 
This implies that if $\bowtie_f(z,z')$ is independent of the choice of path, then $(G,\rho,\tau)$ is orientable. 
We also remark that for any rotation system $(G,\rho,\tau)$ whose resulting surface is orientable, there exists an equivalent rotation system $(G,\rho',0)$ to it. 
Let us see that 
$\bowtie_f (z,z')$ is independent of the choice of path if the underlying surface is orientable: 
assume that $(G,\rho,\tau)$ is orientable. If there exists two paths $p$ and $p'$ with $o(p)=o(p')$, $t(p)=t(p')$ and $\int_p \tau\neq \int_p' \tau$. However 
Lemma~\ref{lem:welldefine} implies that $\int_{p\cup \bar{p'}}\tau=\int_{p}\tau-\int_{p'}\tau=0$, which is contradiction. 
Therefore $\bowtie_f(t(e),t(e'))$ is invariant under the choices from $e\in \delta A(x)$ and $e'\in \delta A_y$ if and only if the underlying surface is orientable. 
\end{proof}
\section{Comfortability (Proof of main theorem)}\label{sect:proofComf}
In this section, we discuss the comfortability under Assumption~\ref{ass:hd}. 
Note that because of the hedgehog boundary condition, the set of tails has a bijection map to the set of bridges $\varphi$ in (\ref{eq:bijection}). 
Let us set $\eta\in \mathbb{C}^{A_{br}}$ by 
\[ \eta(e):=\frac{1}{bc\omega} (Q\bs{\alpha}_{in})(t(\mathrm{is}^\sharp (\bar{e}))). \]
See (\ref{eq:isis}) and Figs.~\ref{fig:isbr}, \ref{fig:qypr} for the difinitions of $\mathrm{is}(e)$ and $\mathrm{is}^\sharp(e)$. 
Here the matrix $Q:=S-dI$ represents the scattering after a quantum walker penetrates the interior at least once. 
By using the notation $\eta$ of such a scattering, the stationary state $\psi_\infty$ is expressed as follows. 
\begin{proposition}\label{prop:stationary}
Under Assumption~\ref{ass:hd}, 
for any bridge arc $e\in A_{br}$, we have 
\[ \psi_\infty(e)=\;\omega\;(\;\eta(e)+(-1)^{\tau(e)}d\eta(\bar{e})\;).  \]
For a facial walk $f$, 
\[ f:=(\;\xi_{0}^{+},e_{0},\;\;\xi_{1}^{-},\xi_1^{+},e_{1},\;\;\xi_{2}^{-},\xi_2^{+},e_{2},\;\dots\\
    \dots,\;\xi_{\kappa-1}^{-},\xi_{\kappa-1}^{+},e_{\kappa-1},\;\;\xi_{0}^{-}\;),\]
where $(\xi_j^-,\xi_j^+)\in \delta A$ and $e_j\in A_{br}$,    
the stationary state at the island arc is     
\[ \psi_\infty(\xi_{m}^{+})=(-1)^{\tau(e_m)} b\eta(e_m),\;\psi_\infty(\xi_{m+1}^{-})=\omega b\eta(e_m),    \]
for any $m\in \mathbb{Z}_{\kappa}$. 
\end{proposition}
\begin{proof}
Let the outflow from the quay $(\xi_{m+1}^-,\xi_{m+1}^+)$ be $\beta_{m+1}$. 
Then $\beta_{m+1}=d\alpha_{m+1}+c\psi_\infty(\xi_{m+1}^-)$, which implies 
\begin{equation}\label{eq:xi-}
\psi_\infty(\xi_{m+1}^-)=(1/c)(S_f-dI_f)\bs{\alpha}_{in}(m+1)=\omega b\eta(e_m). 
\end{equation}
by the definition of the scattering matrix in (\ref{eq:defscat}). 
Lemma~\ref{lem:key} leads 
\begin{equation}\label{eq:xi+} 
\psi_\infty(\xi_{m}^+)=(-1)^{\tau(e_m)}\omega^{-1}\psi_\infty(\xi_{m+1}^-) =(-1)^{\tau(e_m)}b\eta(e_m).
\end{equation}
Under the assumption of $d\in \mathbb{R}$, we have 
\begin{align*}
\begin{bmatrix}
\psi_\infty(e_1^R) \\ \psi_\infty(e_1^L)
\end{bmatrix}
& = \frac{1}{b}
\begin{bmatrix} d' & -b \\ 1 & 0 \end{bmatrix} \begin{bmatrix} -a & 1 \\ -c' & 0 \end{bmatrix}
\begin{bmatrix} \psi_\infty(e_0^R) \\ \psi_\infty(e_0^L) \end{bmatrix} 
\\
&= \frac{1}{b}\begin{bmatrix} 1 & d' \\ d' & 1 \end{bmatrix}
\begin{bmatrix} \omega' & 0 \\ 0 & 1  \end{bmatrix}
\begin{bmatrix} \psi_\infty(e_0^R) \\ \psi_\infty(e_0^L) \end{bmatrix}\\
&= \frac{\omega'}{b}\begin{bmatrix} 1 & d' \\ d' & 1 \end{bmatrix}
\begin{bmatrix} \psi_\infty(e_0^R) \\ \psi_\infty(e_2^L) \end{bmatrix}.
\end{align*}
Here we used (\ref{eq:trandfermat}), (\ref{eq:unitarity}) and (\ref{eq:key}) in the first, the second and the third equalities, respectively.
Then it holds that
\begin{equation}\label{eq:ebri}
\psi_\infty(e_1^R)=\frac{\omega'}{b}(\psi_\infty(e_0^R)+d'\psi_\infty(e_2^L)),\;
\psi_\infty(e_1^L)=\frac{\omega'}{b}(d'\psi_\infty(e_0^R)+\psi_\infty(e_2^L)).
\end{equation}
Remark that 
if we put $e_1^R=e_m\in A_{br}$, then 
\begin{align*}
\psi_\infty(e_0^R) =(-1)^{\tau(e_m)}b\eta(e_m) \text{ and }
\psi_\infty(e_2^L) =(-1)^{\tau(\bar{e}_m)}b\eta(\bar{e}_m)
\end{align*}
hold by (\ref{eq:xi+}). 
Then we obtain
\[ \psi_\infty(e_m)=\omega( \eta(e_m)+(-1)^{\tau(e_m)}d\eta(\bar{e}_m) ) \]
by inserting 
(\ref{eq:xi+}) into (\ref{eq:ebri}), which is the desired conclusion. 
\end{proof}
Let $\sigma:\mathbb{C}^{A_{br}}\to \mathbb{C}^{A_{br}}$ be the flip-flop with the twist $\tau$ such that 
\[ (\sigma \psi)(a)=(-1)^{\tau(a)}\psi(\bar{a})  \]
for any $\psi\in A_{br}$ and $e\in A_{br}$, and set $Q:\mathbb{C}^{A_{br}}\to \mathbb{C}^{A_{br}}$ by 
\[ Q=S-dI. \]  
Here the scattering matrix $S$ is regarded as the operator on $\mathbb{C}^{A_{br}}$ with the bijection map $A_{br}\to \delta X$ by $x= t(\mathrm{is}^\sharp (\bar{e}))\in \delta X$ for any $e\in A_{br}$(: note that the boundary is the hedgehog). 
By using the expression of the stationary state $\psi_\infty$ in Proposition~\ref{prop:stationary}, 
the comfortability $\mathcal{E}$ is described as follows. 
\begin{proposition}\label{prop:comfgen}
Under Assumption~\ref{ass:hd}, 
the comfortability with the inflow $\bs{\alpha}_{in}$ is described by 
\[ \mathcal{E}=\frac{1}{|c|^2} || (S-dI)\bs{\alpha}_{in} ||^2+\frac{1}{2|bc|^2}||(\sigma+dI)(S-dI)\bs{\alpha}_{in}||^2. \]
\end{proposition}
\begin{proof}
The comfortability $\mathcal{E}$ is decomposed into
\[ \mathcal{E}=\mathcal{E}^{island}+\mathcal{E}^{bridge}. \]
Here $\mathcal{E}_{island}=(1/2)\sum_{e\in A_{is}} |\psi_\infty(e)|^2$ and 
$\mathcal{E}^{bridge}=(1/2)\sum_{e\in A_{br}} |\psi_\infty(e)|^2$. 
\begin{enumerate}
\item Island: 
Since every island arc belongs to unique facial walk, $\mathcal{E}^{island}$ can be further decomposed into facial walks by  
\[ \mathcal{E}^{island}=\sum_{f\in F} \mathcal{E}_f^{island}, \]
where $\mathcal{E}_f^{island}=\sum_{e\in f\cap A_{is}} |\psi_\infty(e)|^2$.
By Proposition~\ref{prop:stationary}, the comfortability on the island is deformed by  
\begin{align*}
\mathcal{E}_f^{island} &= \frac{1}{2} \sum_{m=0}^{|f|-1} \left\{\;|(-1)^{\tau(e_m)}b\eta(e_m)|^2+|\omega b\eta(e_m)|^2\;\right\} \\
&= |b|^2 \sum_{m=0}^{|f|-1} \frac{1}{|bc\omega|^2}| (Q_f\bs{\alpha}_f)(m+1) |^2 \\
&= \frac{1}{|c|^2} ||\; (S_f-dI_f)\bs{\alpha}_f \;||^2. 
\end{align*}
Then we have 
\begin{equation}\label{eq:island}
\mathcal{E}^{island}=\sum_{f\in F} \mathcal{E}_f^{island}=\frac{1}{|c|^2}||(S-dI)\bs{\alpha}_{in}||^2.
\end{equation}
\item Bridge: 
By Proposition~\ref{prop:stationary}, we have 
\begin{align}
\mathcal{E}^{bridge} &= \frac{1}{2} \sum_{e\in A_{br}} |\omega (\eta(e)+(-1)^{\tau(e)}d\eta(\bar{e}))|^2 \notag \\
&= \frac{1+d^2}{2} \sum_{e\in A_{br}} |\eta(e)|^2+d\cdot \mathrm{Re}\left[ \sum_{e\in A_{br}}(-1)^{\tau(e)}\eta(e)\overline{\eta(\bar{e})} \right] \notag\\
&= \frac{1}{2}\left\{ (1+d^2)\;||\eta||^2+2d\cdot \mathrm{Re}(\langle \eta,\sigma\eta \rangle) \right\} \notag\\
&= \frac{1}{2} || (\sigma+dI)\eta ||^2 \notag\\
&= \frac{1}{2|bc|^2} || (\sigma+dI)(S-dI)\bs{\alpha}_{in} ||^2. \label{eq:bridge}
\end{align}
\end{enumerate}
Combining (\ref{eq:island}) with (\ref{eq:bridge}), we obtain the desired conclusion. 
\end{proof}
\begin{remark}\label{re,:isocomf}
Let $(G,\rho_1,\tau_1)$ and $(G,\rho_2,\tau_2)$ be rotation systems which are equivalent to each other. 
The scattering matrices and the comfortabilities of each rotation system are denoted by $S_1$, $S_2$ and $\mathcal{E}_1$, $\mathcal{E}_2$, respectively. 
By Proposition~\ref{prop:scatiso}, there exists a diagonal unitary matrix $D$ such that $S_2=D^*S_1D$. Then Proposition~\ref{prop:comfgen} implies that for any inflow $\bs{\alpha}_1$ to the rotation system $(G,\rho_1,\tau_1)$, there exists $\bs{\alpha}_2$ such that $\mathcal{E}_1$ with the inflow $\bs{\alpha}_1$ coincides with $\mathcal{E}_2$ with the inflow $\bs{\alpha}_2$, and the inflow $\bs{\alpha}_2$ is described by $\bs{\alpha}_2=D\bs{\alpha}_1$. 
Thus $\mathcal{E}_1\neq \mathcal{E}_2$ with the same inflow in general. 
However, in the case for a single inflow, that is, $|\supp (\bs{\alpha}_{in})|=1$, then $\mathcal{E}_1=\mathcal{E}_2$, which means that the comfortability with the single inflow is invariant among the equivalent embbedings.
\end{remark}

Concerning the observation in Remark~\ref{re,:isocomf}, now let us set the inflow inserting $1$ from a single tail which is chosen uniformly at random, that is, an inflow $\bs{\alpha}_{in}$ is selected randomly from  
\[ \{ \delta_{e}\;|\;e\in A_{br} \}. \]
Each probability that $\bs{\alpha}_{in}=\delta_e$ ($e\in A_{br}$) is $1/|A_{br}|$. 
We are interested in the average of the comfortability with respect to this randomly setting of the inflow, that is, 
\[\mathbb{E}[\mathcal{E}]= \frac{1}{|A_{br}|} \sum_{e\in A_{br}}\mathcal{E}^{(e)}, \]  
where $\mathcal{E}^{(e)}$ is the comfortability with the inflow $\delta_e$. 
The average of the comfortability is expressed as follows. 
\begin{theorem}\label{thm:comf}
Under Assumption~\ref{ass:hd}, suppose that $\bs{\alpha}_{in}$ is the random variable distributed uniformly in $\{\delta_e\;|\;e\in A_{br}\}$. 
Then the average of the comfortability is expressed by
\begin{align*}
\mathbb{E}[\mathcal{E}] &= \frac{1}{|A_{br}|} \frac{2+|b|^2}{|b|^2}\sum_{f\in F}|f| \frac{1-|a|^{2|f|}}{|\;1-(a\omega)^{|f|}\;|^2}\\
&\qquad + \frac{1}{|A_{br|}}\frac{d}{|b|^2} \sum_{f\in F}\frac{1}{|\;1-(a\omega)^{|f|}\;|^2}\sum_{e\in f\cap \bar{f}} \bigg\{\;(a\omega)^{\mathrm{dist}_f(\bar{e},e)}\;(1-|a|^{2\mathrm{dist}_f(e,\bar{e})})\\
&\qquad\qquad\qquad\qquad\qquad\qquad\qquad\qquad\qquad\qquad
+(\overline{a\omega})^{\mathrm{dist}_f(e,\bar{e})}\;(1-|a|^{2\mathrm{dist}_f(\bar{e},e)})\;\bigg\}.
\end{align*}
Here $|f|$ is the length of the facial walk $f$ in $(G,\rho,\tau)$ and $\mathrm{dist}_f(e,e')$ is the distance from $t(e')$ to $o(e'$ along the facial walk $f$ in $(G,\rho,\tau)$ for any $e,e'\in f$. 
\end{theorem}
\begin{remark}
Theorem~\ref{thm:comfa>0} corresponds to the case for $a>0$ and $\omega =1$ in Theorem~\ref{thm:comf}. 
\end{remark}
\begin{proof}
By Proposition~\ref{prop:comfgen}, it holds that 
\begin{align}
\mathbb{E}[\mathcal{E}] 
&= \frac{1}{|c|^2}\mathbb{E}[\;||Q\bs{\alpha}||^2\;]+\frac{1}{2|bc|^2}\mathbb{E}[\;||(\sigma+dI)Q\bs{\alpha}||^2\;] \notag\\
\notag\\
&= \frac{1}{|c|^2}\frac{1}{|A_{br}|}\mathrm{tr}(Q^*Q)+\frac{1}{2|bc|^2}\frac{1}{|A_{br}|}\mathrm{tr}(\; 
[(\sigma+dI)Q]^*[(\sigma+dI)Q] \;) \notag\\
\notag\\
&= \frac{1}{|A_{br}|}\frac{2+|b|^2}{2|bc|^2}\mathrm{tr}(QQ^*)
+\frac{1}{|A_{br}|}\frac{d}{|bc|^2}\mathrm{tr}(QQ^*\sigma) \label{eq:comfbefore}
\notag \\
\end{align}
Now to describe the above RHS more explicitly, let us compute $QQ^*$ as follows. 
Put $P_f:=P_f(1)$. We remark that $P_f(\omega)=\omega P_f$ for any $f\in F$ under Assumption~\ref{ass:hd}. 
Note that $Q=\oplus_{f\in F}Q_f$, and for each $f\in F$ with $f=(e_0,e_1,\dots,{e_{|f|-1}})$, we have 
\begin{equation}
\sum_{k=0}^{|f|-1}\tau(e_k)=0,
\end{equation}  
which implies $P_f^{|f|}=I_f$. 
Then Theorem~\ref{thm:scattering0} implies that $Q_f$ can be expanded by  
\begin{align*}
Q_f &= bc\omega P_f(I-a\omega P_f)^{-1} \\
&= \frac{bc\omega}{1-(a\omega)^{|f|}}(I+a\omega P_f+\cdots+(a\omega P_f)^{|f|-1}).
\end{align*}
Then we have 
\begin{align}\label{eq:QfQf^*}
Q_fQ_f^* &= \frac{|bc|^2}{|1-(a\omega)^{|f|}|^2}\sum_{k=0}^{|f|-1}\left(\sum_{
{\scriptsize \begin{matrix}
\ell,m\in \mathbb{Z}_{|f|} \\
\text{ with }m-\ell= k
\end{matrix}
}
}(a\omega)^m(\overline{a\omega})^\ell\right)P_f^k.
\end{align}
Here $w_k:=\sum_{\ell,m\in\mathbb{Z}_{|f|},\;m-\ell=k}\;(a\omega)^m(\overline{a\omega})^\ell$ is reduced to 
\begin{align*}
w_k
&= \sum_{j=0}^{|f|-1-k}(\overline{a\omega})^j(a\omega)^{j+k}+\sum_{j=0}^{k-1}(\overline{a\omega})^{|f|-k+i}(a\omega)^i \\
&=\frac{1}{|b|^2}\left\{ (a\omega)^k(1-|a|^{2(|f|-k)})+(\overline{a\omega})^{|f|-k}(1-|a|^{2k}) \right\}. 
\end{align*}
Here we used $|a|^2+|b|^2=1$. 
Note that 
\[(P_f^{k})_{e_\ell,e_m}= 
\begin{cases}
(-1)^{\sum_{j=m+1}^\ell\tau(e_j)} &\text{: $\mathrm{dist}_f(e_\ell,e_m)=k$,}\\
0& \text{: otherwise.}
\end{cases} \]
Here ``$\sum_{j=p}^q\tau(j)$" is defined by the summation of $\tau(j)$'s over the walk along the boundary of the face $f$ from $e_p$ to $e_q$ for $p,q\in \mathbb{Z}_{|f|}$, that is,  
\[\sum_{j=p}^q\tau(j)= \tau(p)+\tau(p+1)+\cdots+\tau(q). \] 
This means that if $p>q$, then $\sum_{j=p}^q\tau(j)=\tau(p)+\tau(p+1)+\cdots+\tau(q+|f|)$. 
Then inserting the above expression for $w_k$ into (\ref{eq:QfQf^*}), we obtain
\begin{multline}\label{eq:QfQf^*2}
(Q_fQ_f^*)_{e_\ell,e_m}=(-1)^{\sum_{k=m+1}^\ell\tau(e_k)}\frac{|b|^2}{|1-(a\omega)^{|f|}|^2}\\
\times \left\{ (a\omega)^{\ell-m}(1-|a|^{2(|f|-(\ell-m))})+(\overline{a\omega})^{|f|-(\ell-m)}(1-|a|^{2(\ell-m)}) \right\},
\end{multline}
for any $\ell,m\in\mathbb{Z}_{|f|}$. 
Note that ``$\ell-m$" and `$|f|-(\ell-m)$" keep the modulus of $\mathbb{Z}_{|f|}$ in the above expression. 
By using (\ref{eq:QfQf^*2}), $\mathrm{tr}(QQ^*)$ is described by 
\begin{align}
\mathrm{tr}(QQ^*) 
&=\sum_{f\in F} \mathrm{tr}\left[ \frac{|b|^2}{|1-(a\omega)^{|f|}|^2}(1-|a|^{2|f|})P_f^{0} \right] \notag \\
&= |b|^2\sum_{f\in F} |f|\frac{1-|a|^{2|f|}}{|1-(a\omega)^{|f|}|^2}, \label{eq:QQ^*}
\end{align}
Let the closed walk on $G^\tau$ from $e$ to $\bar{e}$ along the facial walk $f$ be denoted by $(e_0,e_1,\dots,e_{s-1})$ with $e_0=e$ and $e_{s-1}=\bar{e}$. 
Then $\mathrm{tr}(QQ^*\sigma)$ is described by
\begin{align}
\mathrm{tr}(QQ^*\sigma) 
&=\sum_e\sum_{e'}(QQ^*)_{e,e'}(\sigma)_{e',e}
= \sum_{e}(QQ^*)_{\bar{e},e} (-1)^{\tau(e)}\notag \\
&= \sum_{f\in F} \sum_{e\in f\cap \bar{f}} (Q_fQ_f^*)_{\bar{e},e} (-1)^{\tau(e)}\notag \\
&= \sum_{f\in F} \frac{|b|^2}{|1-(a\omega)^{|f|}|^2} \notag \\
&\quad\times \sum_{e\in f\cap \bar{f}} \left\{ (a\omega)^{\mathrm{dist}_f(\bar{e},e)}(1-|a|^{2\mathrm{dist}_f(e,\bar{e})}) + (\overline{a\omega})^{\mathrm{dist}_f(e,\bar{e})}(1-|a|^{2\mathrm{dist}_f(\bar{e},e)})\right\} (-1)^{\sum_{k=0}^{s-1}\tau(e_k)}.
\label{eq:QQ^*sigma}
\end{align}
Here in the last equality we used (\ref{eq:QfQf^*}) and   $\mathrm{dist}_f(e,\bar{e})=|f|-\mathrm{dist}_f(\bar{e},e)$. 
We remark that $\sum_{k=0}^{s-1}\tau(e)=0$  since $e,\bar{e}$ lie on the same sheet. 
Combining (\ref{eq:comfbefore}) with (\ref{eq:QQ^*}) and (\ref{eq:QQ^*sigma}), we obtain the desired conclusion. 
\end{proof}


\noindent\\
\noindent {\bf Acknowledgments}
We would like to thank Takumi Kakegawa, and Professors Kenta Ozeki, Atsuhiro Nakamoto for the fruitful discussions and suggestions. 
Yu.H. acknowledges financial supports from the Grant-in-Aid of
Scientific Research (C) Japan Society for the Promotion of Science (Grant No.~18K03401, No.~23K03203). 
E.S. acknowledges financial supports from the Grant-in-Aid of
Scientific Research (C) Japan Society for the Promotion of Science (Grant No.~24K06863) and Research Origin for Dressed Photon.


\end{document}